\documentclass[journal]{IEEEtran}
\usepackage{amsmath,amssymb}
\usepackage{graphicx,graphics,color,psfrag}
\usepackage{cite,balance}
\usepackage{caption}
\allowdisplaybreaks
\usepackage{algorithm}
\usepackage{algorithmic}
\usepackage{float}
\usepackage[caption=false,font=normalsize,labelfont=sf,textfont=sf]{subfig}
\usepackage{mdframed} 
\usepackage{framed} 
\usepackage{accents}
\usepackage{amsthm}
\usepackage{bm}
\usepackage{url}
\usepackage[english]{babel}
\usepackage{multirow}
\usepackage{enumerate}
\usepackage{cases}
\usepackage{stfloats}
\usepackage{dsfont}
\usepackage{color,soul}
\usepackage{amsfonts}
\usepackage{cite,graphicx,amsmath,amssymb}
\usepackage{fancyhdr}
\usepackage{hhline}
\usepackage{graphicx,graphics}
\usepackage{array,color}
\usepackage{mathtools}
\usepackage{amsmath}
\usepackage[T1]{fontenc}
\usepackage{float} 
\usepackage{diagbox}
\usepackage{tcolorbox}
\usepackage{float}
\usepackage{siunitx}
\usepackage{lipsum}
\usepackage{adjustbox}
\usepackage{stfloats}

\newtheorem{definition}{\emph{\underline{Definition}}}

\newtheorem{lemma}{\emph{\underline{Lemma}}}
\newtheorem{example}{\emph{\underline{Example}}}

\newtheorem{proposition}{\emph{\underline{Proposition}}}

\newtheorem{remark}{\bf \emph{\underline{Remark}}}

\def\l{\left}
\def\r{\right}
\def\({\left(}
\def\){\right)}

\def\b0{{\mathbf{0}}}

\newcommand{\nn}{\nonumber}

\begin{document}
\captionsetup[figure]{name={Fig.}}
\title{\huge 
Multi-beam Training for Near-field Communications\\ in High-frequency Bands} 
\author{Cong Zhou, Changsheng You,~\IEEEmembership{Member,~IEEE}, Zixuan Huang, Shuo Shi,~\IEEEmembership{Member,~IEEE}, \\ Yi Gong,~\IEEEmembership{Senior Member,~IEEE}, Chan-Byoung Chae,~\IEEEmembership{Fellow,~IEEE}, and Kaibin Huang,~\IEEEmembership{Fellow,~IEEE},
	\thanks{C. Zhou is with the Department of Electrical and Electronic Engineering, Southern University of Science and Technology, Shenzhen 518055, China and also with the School of Electronic and Information Engineering, Harbin Institute of Technology, Harbin, 150001, China (e-mail:  zhoucong@stu.hit.edu.cn). C. You and Y. Gong are with the Department of Electronic and Electrical Engineering, Southern University of Science and Technology, Shenzhen 518055, China (e-mail: youcs@sustech.edu.cn, gongy@sustech.edu.cn). Z. Huang is with the School of Electronics and Communication Engineering, Guangzhou University, Guangzhou, China (email: zxhuang@gzhu.edu.cn). S. Shi is with the School of Electronic and Information Engineering, Harbin Institute of Technology, Harbin, 150001, China (e-mail: crcss@hit.edu.cn). C. B. Chae is with the School of Integrated Technology, Yonsei University, Seoul, 03722, Korea (email: cbchae@yonsei.ac.kr). K. Huang is with Department of Electrical and Electronic Engineering, The University of Hong Kong, Hong Kong, China (email: huangkb@eee.hku.hk).}
}
\maketitle
\begin{abstract}
	In this paper, we study  efficient  \emph{multi-beam} training design for \emph{near-field} communications to reduce the beam training overhead of conventional single-beam training methods. In particular, the array-division based multi-beam training method, which is widely used in far-field communications, cannot be directly applied to the near-field scenario, since different sub-arrays may observe different user angles and there exist coverage holes in the angular domain. To address these issues, we first devise a new near-field multi-beam codebook by sparsely activating a portion of antennas to form a  \emph{sparse linear array} (SLA), hence generating multiple beams simultaneously by effective exploiting the near-field \emph{grating-lobs}. Next,  a two-stage near-field beam training method is proposed, for which several candidate user locations are  identified firstly based on multi-beam sweeping over time, followed by the second stage to further determine the true user location with a small number of single-beam sweeping.  Finally, numerical results show that our proposed multi-beam training method significantly reduces the beam training overhead of conventional single-beam training methods, yet achieving comparable  rate performance in data transmission.
\end{abstract}

\section{Introduction}
\emph{Extremely large-scale array} (XL-array) has emerging as a promising technology to significantly enhance the spectral efficiency and spatial resolution of future wireless networks, via increasing the number of antennas by another order-of-magnitude \cite{you2024next}. In addition, the substantially increased array aperture also brings a fundamental change to the radio propagation modelling, shifting from the conventional far-field communications with planar wavefronts to the  \emph{near-field} counterpart with spherical wavefronts (SWs) \cite{lu2023tutorial}. This SW channel characteristic not only provides appealing opportunities such as near-field beam-focusing and joint resource management in the angle-and-range domains \cite{zhang20236g}, but also imposes new challenges in e.g., the beamforming design and channel estimation \cite{Cui2022channel}. In this paper, we consider a multi-user XL-array communication systems and study efficient near-field multi-beam training design.
\subsection{Prior Work}
For near-field communications, channel state information (CSI) acquisition is indispensable yet practically challenging, due to the larger number of antennas and the new SW channel model. The existing work on this area can be largely classified into two categories, namely,  explicit channel estimation and inexplicit beam training.
\subsubsection{Near-field channel estimation}
For high-frequency band systems such as millimeter-wave (mmWave) and terahertz (THz), compressed sensing (CS) techniques have been widely used to estimate the sparse channels in the far-field multi-input multi-out (MIMO) communication systems. These techniques, however, may not be directly applied to the near-field scenario, since the SW model does not exhibit spatial sparsity in the angular domain.
To address this issue, the authors in \cite{Cui2022channel} proposed a \emph{polar}-domain codebook to effectively represent the near-field channels such that the SW channels admit spatial sparsity in the polar domain. As such, the general multi-path near-field channels in narrow-band systems can be estimated by using customized CS techniques such as orthogonal matching
pursuit (OMP) \cite{Cui2022channel}.
Further, the authors in \cite{zhang2023near} proposed an angular-domain dictionary with a range parameter and iteratively obtained the angular-domain sparse vector and range information, hence significantly reducing the codebook size in \cite{Cui2022channel}.
Moreover, for mixed line-of-sight (LoS)/non-LoS (NLoS) near-field channels, a customized two-stage near-field channel estimation method was proposed in \cite{lu2023near}, which separately estimates the LoS and NLoS channel components by using OMP algorithms.
In addition, the near-field channel estimation methods in narrow-band systems were further extended to the wide-band scenario, where a bilinear pattern detection method was exploited in \cite{cui2023nearwideband}. However, the above works have mostly considered the near-field environmental scatterers only, while far-field scatterers may also exist in the general case. For this scenario, the authors in \cite{wei2021channel} developed an efficient hybrid near- and far-field channel estimation method, for which the far-field and near-field path components were respectively estimated by using OMP techniques. In addition to the SWs, spatial non-stationarity is another main channel feature for near-field XL-array communications, where only a portion of antennas are visible to the users due to environmental scatterers and blockages. 
Specifically, the authors in \cite{chen2023non} proposed to exploit spatial non-stationary effect to reduce pilot overheads and converted the non-stationary channel to a series of spatial stationary channels, which were estimated by CS based techniques.
\subsubsection{Near-field beam training}
Alternatively, beam training is another efficient method for acquiring CSI, which aims to select the best beamforming vector from a predefined beam codebook \cite{zhang2023nearYou, zhang2022dual}. This method has been widely used in practice and documented in IEEE 802.15.3c and IEEE 802.11ad protocol \cite{xiao2016hierarchical, hassanieh2018fast}, since it usually requires power measurements and comparisons only, thus avoiding demanding computational complexity in channel estimation methods.
Moreover, beam training methods are effective in both the low and high signal-to-noise ratio (SNR) regimes, while several channel estimation methods (e.g., least square (LS)) may suffer degraded performance in the low-SNR regime \cite{zheng2022survey}.
However, the existing beam training methods for far-field communications cannot be directly applied to the near-field scenario due to the channel model mismatch.
To tackle this issue, the authors in \cite{cui2022near} proposed to leverage the polar-domain codebook for  the single-beam training based on the exhaustive search, which, however, incurs prohibitively high beam training overhead that scales with the product of numbers of sampled angles and ranges. 
To reduce the beam training time, a variety of more efficient near-field beam training methods have been recently proposed in the literature \cite{Zhang2022fast, wu2024near,wu2023two}.
For example, the authors in \cite{Zhang2022fast} developed a novel two-phase beam training method that first estimates the user angle using the discrete Fourier transform (DFT) codebook and then resolves the user range using the polar-domain codebook.
This method was further improved in \cite{wu2024near} where both the user angle and range are estimated jointly based on the DFT codebook.
Moreover, hierarchical near-field beam training methods were proposed in \cite{lu2023hierarchical, wu2023two, chen2023hierarchical} to achieve lower beam training overhead that essentially finds the coarse user angle (and/or range) firstly by using wide beams and then gradually resolves fine-grained user angle-and-range with narrow beams.

Compared with the single-beam training, multi-beam training design is capable of generating multiple beams simultaneously, hence effectively reducing the beam training overhead. Specifically, for far-field communication systems, the array-division method has been widely used for designing multi-beam codebooks, where each sub-array is responsible for beam sweeping in an angular subspace only, under the assumption that all sub-arrays share the same user angle. 
Then, the true user angle can be determined by comparing the received powers at the users over time based on different methods such as the cross validation \cite{you2022fast} and random hashing (RH) \cite{hassanieh2018fast}. However, these multi-beam training methods may not be applicable to the near-field communication systems, since 1) users generally locate at different regions of different sub-arrays, 2) different sub-arrays may observe different user angles, and 3) there exist uncovered areas in the polar domain (which will be detailed in Section \ref{Sec:Far-field}). Another multi-beam training design leverages the near-field beam-split effect in wide-band communications by using time-delay (TD) beamforming to control the beams at different frequencies, focusing them on different locations \cite{cui2023rainbow}.
However, this method relies on wide-band systems only.
These thus motivate the current work to design a new and efficient multi-beam training method tailored to near-field narrow-band communication systems.
\subsection{Contributions and Organizations}
In this paper, we consider a multi-user near-field communication systems as shown in Fig. \ref{fig:systemModel}, where an XL-array base station (BS) servers multiple single-antenna users under the LoS-dominant channel model. To reduce the design complexity of hybrid beamforming, we first design the analog beamforming using the beam training method and then devise the digital beamforming based on estimated effective channels. Different from existing works that mostly considered near-field single-beam training, we aim to design a new multi-beam codebook and an efficient beam training method to reduce the beam training overhead.
The main contributions are summarized as follows.
\begin{itemize}
	\item First, we design a new near-field multi-beam codebook by sparsely activating a portion of antennas to form an \emph{SLA}. To this end, we first characterize the near-field beam pattern for the effective SLA and show that each beam codeword steering towards a specific location can essentially generate multiple beams simultaneously, among which one is the main-lobe and the others are grating-lobs. As such, only a subspace in the polar domain (corresponding to the main-lobe) is sampled for establishing the near-field multi-beam codebook.
	\item Second, we propose a novel two-stage multi-beam training method, for which several candidate user locations are  identified firstly based on the multi-beam sweeping over time, followed by the second stage to further determine the true user location with a small number of single-beam sweeping.
	Moreover, we show that this method is applicable to the far-field scenario as well as the general multi-path channel setup.
	\item Finally, we provide extensive numerical results to demonstrate the effectiveness of our proposed near-field multi-beam training method and shed useful insights.
	It is shown that the proposed multi-beam training method can achieve close rate performance to the exhaustive-search method under various simulation setups, yet requiring much lower training overhead.
	Moreover, by controlling the antenna activation sparsity, there generally exists a trade-off between reducing the beam training overhead and increasing the beam training accuracy.  
	
\end{itemize}

The remainder of this paper is organized as follows.
In Section \ref{Sec:System model and problem formulation}, we present the system model and problem formulation.
In Section \ref{Sec:Far-field}, we discuss the main issues when the existing array-division method is used for near-field multi-beam training. Then, a new array-sparse-activation method is proposed in Section \ref{Sec:Near-field multi-beam codebook design} for designing a customized near-field multi-beam codebook and  
 a two-stage near-field multi-beam training method was proposed in Section \ref{Sec:scheme1}.
Finally, numerical results are presented in section \ref{Sec:numericalResults} to demonstrate the effectiveness of the proposed method, followed by the conclusions made in Section \ref{Sec:Conclusion}. 

\textit{Notations}: Lower-case and upper-case boldface letters  denote vectors and matrices, respectively. Moreover, calligraphic letter is used to represent set. For vectors and matrices, the symbol $ (\cdot)^{H} $ denotes the conjugate transpose operation. The notations $ \left|\cdot\right| $ and $ \left\lVert \cdot\right\lVert $ represent the absolute value of a numerical entity and the $ \ell_{2} $ norm, respectively.
The symbol $\mathbf{I}_K $ denotes a $ K $-dimensional identity matrix.

\section{System Model and Problem Formulation}
\label{Sec:System model and problem formulation}
We consider a multi-user XL-array downlink communication system as shown in Fig.~\ref{fig:systemModel}, where a BS equipped with an $ N $-antenna uniform linear array (ULA) serves $ K $ single-antenna users. 
For ease of exposition, we assume that $ N $ is an odd number.

\subsection{System Model}
\underline{\textbf{Near-field channel model:}} We assume that the XL-array is located at the $y$-axis and centered at the origin. Specifically, each antenna of the XL-array is located at ($0, nd_{0}$), where $ n \in \mathcal{N} \triangleq \{0,\pm 1,\cdots, \pm \frac{N-1}{2}\} $ denotes the antenna index and $ d_{0}$ denotes the inter-element spacing.
For the dense ULA, we have $ d_{0} = \frac{\lambda}{2}$, where $ \lambda $ represents the carrier wavelength.
Moreover, all users are assumed to be located in the Fresnel near-field region of the XL-array, for which the BS-user range $ r_{k} $ is within $Z_{\rm F}\le r_{k} \le  Z_{\rm R}$, where $ Z_{\rm F}\triangleq\max{\{ d_{\rm R} , 1.2D \} }   $ and $ Z_{\rm R} \triangleq\frac{2D^2}{\lambda}$ denote the \emph{Fresnel distance} and the \emph{Rayleigh distance}, respectively \cite{selvan2017fraunhofer} with $ D = (N-1)d_{0} $ representing the array aperture.
Moreover, $ d_{\rm R} $ represents the boundary between the reactive and radiating near-field regions, which is verified to be several wavelengths \cite{ouyang2024impact}.
Therefore, the Fresnel distance can be simplified by $ Z_{\rm F} = 1.2D $.
For example, when $N=257$ and $f = 30$ GHz, the Rayleigh and Fresnel distances are $328$ m and $1.536$ m, respectively. 
This indicates that in cellular networks, the users are very likely to be located in the radiative near-field region.
It is worth noting that in the Fresnel region, given $r_k >1.2D $, the channel amplitude variations over the antennas are negligible, while the phase variations are non-linear. Based on the above, the general  multi-path channel from the XL-array to user $ k $ can be modeled based on  \emph{uniform} SWs (USWs) \cite{zhang2023joint, zhang2023mixedYou}
\begin{equation}
	\label{Eq:nf-channel}
	\mathbf{h}^H_{k} = \sqrt{N}\beta_{k} \mathbf{b}^{H}(r_{k}, \theta_{k}) + \sum_{\ell=1}^{L_k} \sqrt{\frac{N}{L_k}}\beta_{k,\ell} \mathbf{b}^{H}(\bar{r}_{k,\ell}, \bar{\theta}_{k,\ell}),
\end{equation}
which includes one LoS path and $L_k$ non-LoS (NLoS) paths.
Herein, the parameters $ \beta_{k} $ (or $ \beta_{k,\ell} $), $ {r}_{k} $ (or $ \bar{r}_{k,\ell} $) and $ \theta_{k} $ (or $ \bar{\theta}_{k, \ell} $) respectively denote the complex LoS (or NLoS) path gain, and the spatial angle and range of the signal path. 
Specifically, the complex gain of the LoS path can be modeled as $ \beta_{{k}}=\sqrt{\frac{\kappa_k}{\kappa_k+1}} \frac{\sqrt{\alpha_0}}{r_{{k}}} e^{-\frac{\jmath 2 \pi r_{{k}}}{\lambda}} $, where $ {\kappa_k} $ and  $\alpha_0$  represent the Rician factor and the reference channel power gain at a range of $1$ m, respectively.
Moreover, the complex gain of the $ \ell $-th NLoS path is assumed to follow the Gaussian distribution, i.e., $\beta_{k,\ell} \sim  \mathcal{C N}(0, \sigma_{k,\ell}^2) $, where $ \sigma_{k,\ell}=\sqrt{\frac{1}{\kappa_k+1}} \frac{\sqrt{\alpha_0}}{r_{{k}}}$ \cite{10123941,LuCommunicating2022}.
In this paper, we mainly consider the near-field communication scenarios in high-frequency bands such as mmWave and even THz. 
In these scenarios, the NLoS channel paths exhibit negligible power due to the severe path-loss and shadowing, while the case with comparable multi-path components will be discussed in Section \ref{Sec:scheme1} and the numerical results under the general Rician fading channel model will be presented in Section \ref{Sec:numericalResults}.\footnote{For near-field beam training, the proposed method aims to design the phase shifts for analog beamforming, which is still applicable to the case of non-uniform SW model with non-negligible amplitude variations over antennas, albeit with different power gains.}  
In the following, we mainly consider the LoS channel path for which the BS-user channel can be approximated as $ \mathbf{h}^H_{k}\approx \sqrt{N}\beta_{k} \mathbf{b}^{H}(r_{k}, \theta_{k}),\forall k \in \mathcal{K} \triangleq \{1,2,\cdots,K\} $.
Based on the USW model, the normalized near-field channel steering vector, $\mathbf{b}^{H}(r_{k},\theta_{k}) \in \mathbb{C}^{1 \times N}$, can be modeled as \cite{Zhang2022fast}

\begin{figure}[t]
	\includegraphics[width=9cm]{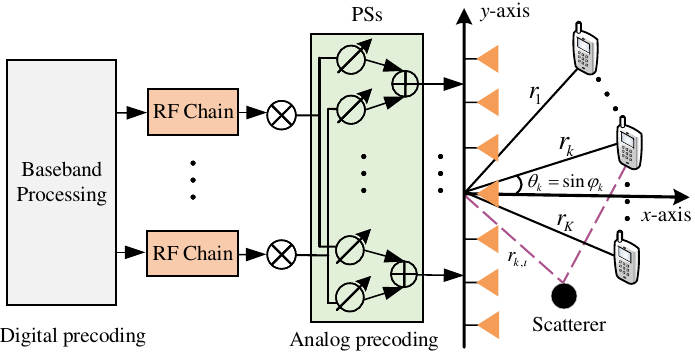}
	\centering
	\caption{System model with hybrid beamforming architecture.}
	\label{fig:systemModel}
	\vspace{-10pt}
\end{figure}
\begin{equation}\label{near_steering}
	\left[\mathbf{b}^H\left(r_{k},\theta_{k}\right) \right]_{n} = \frac{1}{\sqrt{N}}e^{-\frac{\jmath 2 \pi (r_{k,n}-r_k)}{\lambda}}, \forall n\in \mathcal{N}, \forall k \in \mathcal{K}, 
\end{equation}
where $r_{k,n}=\sqrt{r_{k}^2+ {n}^2d_0^2-2r_{k}\theta_{k} nd_0}$ denotes the range between the $ n $-th antenna and the $ k $-th user. 
Moreover, $\theta_{k} \triangleq \cos \alpha_k \in[-1,1]$ represents the spatial angle at the BS with $ \alpha_k $ denoting the physical angle-of-departure (AoD) from the BS center to user $ k $.
By means of Fresnel approximation, $ r_{k,n} $ can be approximated as \cite{Cui2022channel}
\begin{equation}
	\label{eq:fresnelApproximation}
	r_{k,n} \approx r_{k}-nd_0 \theta_{k}+\frac{n^2 d_0^2 (1- \theta_{k}^2)}{2 r_{k}},
\end{equation}
which is accurate when the phase error caused by Fresnel approximation does not exceed $ \pi/8 $~\cite{liu2023near}.

\underline{\textbf{Signal model:}} 
We consider the cost-effective hybrid beamforming architecture for the XL-array, which comprises of $ N_{\rm RF} $ radio frequency (RF) chains with $ K \le N_{\rm RF} \le N $.
To minimize the energy consumption, we consider a practical setup with $ N_{\rm RF} = K $ \cite{sun2019beam}.
Let $\mathbf{x}\in \mathbb{C}^{K \times 1}$ denote the transmitted signals for the $ K $ users with $ \mathbb{E}\left[\mathbf{x}\mathbf{x}^{{H}}\right]=\frac{P_{\rm {tol }}}{K} \mathbf{I}_K $, where $ P_{\rm tol} $ represents the total transmit power of the BS.
Moreover, $\mathbf{F}_{\rm RF}\in \mathbb{C}^{N \times K}$ and $\mathbf{F}_{\rm BB}\in \mathbb{C}^{K \times K}$ respectively denote the analog and digital beamforming matrices.
Based on the above, the received signal at user $ k $ is given by
\begin{align}\label{Eq:general_Sig}
	y_{k} \approx \tilde{y}_k\triangleq
	%    \mathbf{h}^H_{\rm near}\mathbf{v}x+z
	\sqrt{N}\beta_{k}\mathbf{b}^{H}(r_{k},\theta_{k})\mathbf{F}_{\rm RF}\mathbf{F}_{\rm BB}\mathbf{x}+z_k,
\end{align}
where $z_k \sim  \mathcal{C N}\left(0, \sigma^2\right)$ is the received additive white Gaussian noise (AWGN) at the $ k $-th user and $ \sigma^2 $ represents the noise power. 
Specifically, the entries of $ \mathbf{F}_{\mathrm{RF}} $ satisfy $ |[\mathbf{F}_{\mathrm{RF}}]_{k_1,k_2}| = 1, \forall k_1, k_2 \in \mathcal{K} $ and the (normalized) power constraint is given by $ \left\|\mathbf{F}_\mathrm{RF}\mathbf{f}_{\mathrm{BB},k} \right \|^{2}_{F}  = 1 $ \cite{liu2024near}.
As such, the achievable rate of the $ k $-th user in bits/second/hertz (bps/Hz) is obtained as
\begin{equation}\label{eq11}
R_k=\log_2\left(1+\frac{\frac{P_{\rm {tol }}}{K}\left|\mathbf{h}_k^{H} \mathbf{F}_{\mathrm{RF}} \mathbf{f}_{\mathrm{BB}, k}\right|^2}{\frac{P_{\rm {tol }}}{K}\sum_{ i \neq k}\left|\mathbf{h}_k^{H} \mathbf{F}_{\mathrm{RF}} \mathbf{f}_{\mathrm{BB}, i}\right|^2+\sigma^2}\right),
\end{equation}
where $ \mathbf{f}_{\mathrm{BB}, k} $ represents the $ k $-th column of matrix $ \mathbf{F}_{\mathrm{BB}} $.
\subsection{Problem Formulation}
We aim to maximize the system sum-rate $ R_{\rm sum} = \sum_{k=1}^{K}R_k $. 
Specifically, for the analog beamforming, we consider the practically used \emph{codebook-based} design \cite{he2017codebook}, for which $ \mathbf{f}_{\mathrm{RF}, k} $ for user $ k $ is chosen from a pre-defined codebook, denoted by $ \mathcal{W} $.
Based on the above, the sum-rate maximization problem can be formulated as
\begin{equation}
	\begin{aligned}
		({\bf P1}):&\!\!\!\max\limits_{\mathbf{F}_{\mathrm{RF}}, \mathbf{F}_{\mathrm{BB}}}\! \sum\limits_{k = 1}^{K}\! \log_2\!\left(\!1\!\!+\!\!\frac{\frac{P_{\rm {tol }}}{K}\left|\mathbf{h}_k^{H} \mathbf{F}_{\mathrm{RF}} \mathbf{f}_{\mathrm{BB}, k}\right|^2}{\frac{P_{\rm {tol }}}{K}\!\!\sum_{ i \neq k}\!\!\left|\mathbf{h}_k^{H} \mathbf{F}_{\mathrm{RF}} \mathbf{f}_{\mathrm{BB}, i}\right|^2\!\!\!\!+\!\!\sigma^2}\right) \\
		&~~~{\text{s}}{\text{.t}}{\rm{. }} ~~~ \mathbf{f}_{\mathrm{RF},k} \in \mathcal{W},\\
		&~~~\quad\quad\left \| \mathbf{F}_\mathrm{RF}\mathbf{f}_{\mathrm{BB},k} \right \|^{2}_{F}  = 1, \forall k \in \mathcal{K}.
	\end{aligned}
\end{equation}
Due to the coupling between the analog beamformer $ \mathbf{F}_{\mathrm{RF}} $ and digital beamformer $ \mathbf{F}_{\mathrm{BB}} $, it is generally difficult to obtain an optimal solution to Problem $ ({\bf P1}) $.
To tackle this difficulty, we consider a practical and low-complexity hybrid beamforming method \cite{liu2023near}, where the analog beamformer $ \mathbf{F}_{\mathrm{RF}} $ is designed to steer  $ K $ beams towards the targeted $ K $ users for maximizing the received signal power at each individual user.
As such, the effective channel accounting for the analog beamforming can be expressed as $ \mathbf{g}_k^{H} = \mathbf{h}_k^{H}\mathbf{F}_{\mathrm{RF}}, \forall k\in \mathcal{K} $.
Then, the digital beamforming can be  devised to deal with the residual inter-user interference by using e.g., the zero-forcing (ZF) or  minimum mean-squared error (MMSE) beamformer.
For example,  the normalized ZF beamforming vector, denoted by $ {\mathbf{f}}_{\mathrm{BB},k}^{\rm ZF} $, is given by $ {\mathbf{f}}_{\mathrm{BB},k}^{\rm ZF} =  \bar{\mathbf{f}}_{\mathrm{BB},k}^{\rm ZF} /\| \bar{\mathbf{f}}_{\mathrm{BB},k}^{\rm ZF}  \| $, where
\begin{equation}
	\bar{\mathbf{f}}_{\mathrm{BB},k}^{\rm ZF} = \left(\mathbf{I}_K-\bar{\mathbf{A}}_k\left(\bar{\mathbf{A}}_k^H \bar{\mathbf{A}}_k\right)^{-1} \bar{\mathbf{A}}_k^H\right) \mathbf{g}_k, \forall k \in \mathcal{K},
\end{equation}
with $ \bar{\mathbf{A}}_k=\left[\mathbf{g}_1, \cdots, \mathbf{g}_{k-1}, \mathbf{g}_{k+1}, \cdots, \mathbf{g}_K\right] $.
Similarly, for the MMSE beamforming, the digital beamforming vector $ \bar{\mathbf{f}}_{\mathrm{BB},k}^{\rm MMSE} $ can be obtained as ${\mathbf{f}}_{\mathrm{BB},k}^{\rm MMSE} = \mathbf{B}_k^{-1} \mathbf{g}_k/\left\|\mathbf{B}_k^{-1} \mathbf{g}_k\right\|, \forall k \in \mathcal{K}$
where $ \mathbf{B}_k=\sum_{i \neq k}^K \frac{P_{\rm tol}}{K\sigma^2} \mathbf{g}_i \mathbf{g}_i^H+\mathbf{I}_K $.

Therefore, the analog beamforming optimization problem is equivalent to selecting $ K $ beam codewords to maximize the beam gain of each user. The corresponding optimization problem can be formulated as
\begin{equation}
	\begin{aligned}
		({\bf P2}):~~~~ & \max\limits_{\mathbf{f}_{\mathrm{RF},k}} \left| \mathbf{b}^H(r_k,\theta_k){\mathbf{f}_{\mathrm{RF},k}} \right| \\
		&~~{\text{s}}{\text{.t}}{\rm{. }} ~ \mathbf{f}_{\mathrm{RF},k} \in \mathcal{W},
	\end{aligned}
\end{equation}
for each user $k$. Problem $ ({\bf P2}) $ is essentially a \emph{near-field beam-training} problem.
Although there are several single-beam training methods proposed in the existing literature (e.g., \cite{Cui2022channel,Zhang2022fast}),
they usually require a large number of training symbols since only one beam is generated at each time.
For instance, consider the case where the BS has 1025 antennas and there are 5 range samples in the  polar-domain codebook \cite{Cui2022channel}. The overhead of the exhaustive-search method \cite{Cui2022channel} and the two-phase beam training method \cite{Zhang2022fast} is up to 5125 and 1030, respectively, which is unaffordable in practice.
To address this issue, in the sequel, we first point out the main issues when applying the existing multi-beam training methods in far-field communication scenarios to the new near-field counterpart.
Then,  an efficient \emph{array-sparse-deactivation} based method is proposed to design fast near-field multi-beam training.

\section{Far-field  Multi-beam Training Methods}\label{Sec:Far-field}
In this section, we first introduce the existing array-division based multi-beam training method. Then, we discuss its main limitations and challenges when applied to the near-field scenario.

\begin{figure}[!t]
	\centering
	\captionsetup[subfloat]{labelfont=rm, format=plain,labelformat=empty}	
	\subfloat[{\small {\rm (a) Issues 1 and 2}.}]{\includegraphics[width=8.5cm]{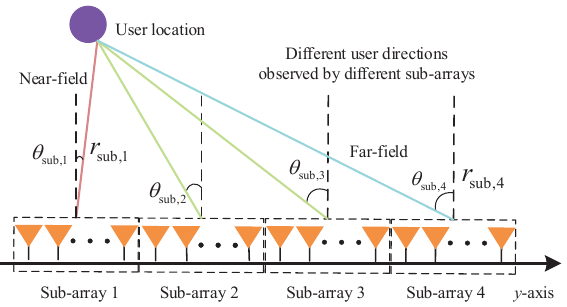} 
	\label{fig:SubArray}}
	\hfil
	\subfloat[{\small {\rm (b) Issue 3: Coverage holes.}}] {\includegraphics[width=8.5cm]{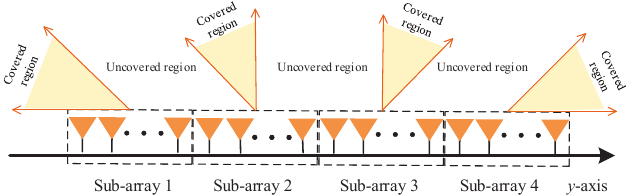} \label{fig:SubArrayHole}}
	\caption{{Issues of the existing array-division method when applied to the near-field beam training.}}
	\label{fig:issue}
	\vspace{-0.5cm} 
\end{figure}
\subsection{Array-division based Multi-beam Training}
For far-field communications, a widely used method to generate multiple beams concurrently is dividing the whole array into $ \tilde{M} $ sub-arrays, each steering its beam towards one specific angle with $ \rho  =  {N}/{\tilde{M}} $ antennas.
Mathematically, let $ \theta_i, i\in\mathcal{I} \triangleq \{1,2,\cdots,\tilde{M}\}$ denote the targeted beam angles.
The multi-beam codeword $ \tilde{\mathbf{w}}$ can be designed as \cite{you2022fast}
\begin{equation}
	[\tilde{\mathbf{w}}]_{(i-1)\rho +1:i\rho } = e^{\jmath\pi(i-1)\rho \theta_i}e^{-\jmath\pi(0:\rho -1)\theta_i},\forall i \in \mathcal{I}.
\end{equation}
Then, the multi-beam codebook for beam training can be generated based on different methods such as RH \cite{hassanieh2018fast} and equal interval multi-arm beam (EIMB) \cite{you2022fast}.

\subsection{Key Issues When Applied to Near-field Scenarios}
Note that the array-division based multi-beam training method generally requires three assumptions:

\textbf{Assumption 1:} The served user is located in the far-field region of each sub-array.

\textbf{Assumption 2:} There is a common angle from the user to all sub-arrays.

\textbf{Assumption 3:} The angle sweeping of all sub-arrays can cover the entire angular region. 

The first assumption is to enforce the effectiveness of the user angle search, while the second assumption is to ensure that the user angle can be jointly determined by the collective beam determination over all sub-arrays. 
Moreover, the third assumption is to ensure that the far-field multi-beam codebook can cover the user angle.
However, these assumptions may \emph{not} practically hold in the new near-field scenario, hence resulting in considerable beam training accuracy loss if directly applying this method. The detailed reasons are elaborated below.
\begin{itemize}
	\item[1.] \textbf{Users generally locate at different regions of different sub-arrays:} 
	Without the user location information, it is generally difficult to partition the XL-array such that the served user is guaranteed to fall into the far-field region of all sub-arrays, thus making Assumption 1 invalid. For example, as illustrated in Fig. \ref{fig:issue}(a), when the user is located near sub-array 1, while very far away from sub-array 4, we have $ r_{{\rm sub},1} < r_{{\rm sub},4}$ and $ \theta_{{\rm sub},1} \ll \theta_{{\rm sub},4}$. This indicates that the user is more likely to be situated in far-field of sub-array 4 while the near-field of sub-array 1 \cite{wu2023enabling}.
	To address this issue, a straightforward approach is by further reducing the number of antennas in each sub-array to ensure that users are more likely located in the far-field region of each sub-array. However, this will render broader beams and significantly reduced sub-array gains, resulting in lower spatial resolution and accuracy \cite{xiao2016hierarchical}.
	For instance, when $ {\tilde M} = 8 $, the success identification rate of the multi-beam training method based on sub-arrays in the far-field is less than 50\%, even in the high-SNR regime \cite{you2022fast}.
	\item[2.] \textbf{Different sub-arrays generally observe different user angles:}
	Even if the user is in the far-field of each sub-array, different sub-arrays may have different angles with the user due to the large XL-array aperture size, especially for sub-arrays at the two ends. 
	Therefore, Assumption 2 may fail in the near-field for the array-division method.
	In Fig. \ref{fig:issue}(a), we present different angles observed by four sub-arrays.
	It is observed that the angles from the user to each sub-array center are generally not the same, since $ \theta_{{\rm sub},1} \ll \theta_{{\rm sub},2} < \theta_{{\rm sub},3} < \theta_{{\rm sub},4}$.
	This thus renders the collective beam determination over all sub-arrays inapplicable.
	\item[3.] \textbf{There exist coverage holes in the angular domain:}
	When the whole array is divided into $ \tilde{M} $ sub-arrays, each sub-array is responsible for sweeping $ {1}/{\tilde{M}} $ of the spatial domain.
	This method can jointly cover the entire angular domain in the far-field case.
	However, in the near-field, the array-division method inevitably causes numerous coverage holes in the angular domain, as illustrated in Fig. \ref{fig:issue}(b).
	It is observed that the angular coverage of the four sub-arrays only occupy part of the angular domain, while there exists a large portion of uncovered areas distributed in the angular domain. 
	This indicates that when the users are located at the uncovered holes, the array-division method will fail to estimate the user locations, hence making  Assumption 3 inapplicable.
\end{itemize}

To address the above issues, we propose a new array-sparse-activation method in the next section to construct multiple beams for near-field communications.

\section{Near-field Multi-beam Codebook Design}
\label{Sec:Near-field multi-beam codebook design}
In this section, we first study the beam pattern of a new array-sparse-activation method.
Then, we show an interesting result that the proposed method can be used for designing the  multi-beam codebook for near-field communication systems.

To this end, we first give the main definitions as follows, including the near-field beam pattern, beam-width and beam-depth \cite{zhou2024sparse}.

\begin{definition}[Beam pattern]\label{De:beam_gain}
	\emph{For a polar-domain codeword $ \mathbf{w} = \mathbf{b}(r_{0}, \theta_{0}) $, the beam pattern characterizes the (normalized) beam power at arbitrary user locations $ (r, \theta) $. Mathematically, the near-field beam pattern of the codeword $ \mathbf{w} $ is 
		\begin{equation}
			f( r, \theta;\mathbf{w} )\triangleq|\mathbf{b}^H(r,\theta) \mathbf{w}|, \forall r, \theta.
		\end{equation} }
\end{definition}

\begin{definition}[Beam-width]\label{Def:width}
	\rm For a near-field codeword $ \mathbf{w} = \mathbf{b}(r_{0}, \theta_{0}) $, the null-to-null beam-width is the spatial angular width in the \emph{user-ring} (which is defined as $ \frac{1- \theta^2}{r} = \frac{1-\theta_0^2}{r_0} $ \cite{zhou2024sparse}), from which the magnitude of  beam pattern decreases to zero (respectively denoted as $ \theta_{\rm right} $ and $ \theta_{\rm left} $) away from the highest power, i.e.,  
	\begin{equation}
		{\rm BW}_{\mathbf{w}} \triangleq \left| \theta_{\rm right} - \theta_{\rm left} \right|.
	\end{equation}
\end{definition}

\begin{definition}[Beam-depth]\!\!\label{Def:Depth}
	\rm For a near-field codeword $ {\mathbf{w}} $ with beam angle $ \theta $, the 3-dB beam-depth characterizes the range interval $r\in [r_{\rm small}, r_{\rm large}]$ for which it satisfies \cite{9723331} $$ f^2({r},{\theta}; \mathbf{w})\le  0.5 \max\limits_{r \in  (Z_{\rm F},Z_{\rm R})}\{f^2({r},{\theta};\mathbf{w})\}.$$  Mathematically, the beam-depth is given by
	%\vspace{-0.2cm} 
	\begin{equation}
		{\rm BD}_{\mathbf{w}}  \triangleq \left| r_{\rm large} - r_{\rm small} \right|.
	\end{equation}
\end{definition}

\subsection{Proposed Array-sparse-activation Method}
In contrast to exiting multi-beam codebooks based on sub-arrays, we aim to design a new near-field codebook based on array sparse activation.
Specifically, as shown in Fig. \ref{fig:Sparse activation difference}, for the XL-array with $ N $ antennas, we only \emph{activate} $ Q $ antennas uniformly with ($ M-1 $) deactivated antennas in between.
To this end, the XL-array becomes an effective \emph{SLA}.
For the XL-array beamforming, this method is equivalent to uniformly sampling the beamforming vector $ \mathbf{w} $ with an interval of $ M $, while the positions unsampled are padded with zeros.
As such, the XL-array beamforming vector based on array sparse activation, denoted by $ \mathbf{v} \in \mathbb{C}^{N \times 1}$, can be expressed as
\begin{equation}
	\label{Eq:samped vector}
	\mathbf{v}^{T} = \left[[\mathbf{w}]_{-\frac{N-1}{2}},\underbrace{0,\cdots,0,}_{M-1}[\mathbf{w}]_{M-\frac{N-1}{2}},\cdots,[\mathbf{w}]_{\frac{N-1}{2}}\right].
\end{equation}
%We define $ M $ as activation {\emph activation interval}, which characterizes the sparsity level of the activated antennas.
It can be verified that the number of non-zero elements in $ \mathbf{v} $ is $ Q =  \frac{N-1}{M}+1 $ (assuming $ Q $ being an integer for convenience).
%Then, we rearrange the non-zero elements in $ \mathbf{v} $ into a vector $ \bar{\mathbf{v}} \in \mathbb{C}^{Q \times 1} $.
As such, the beam pattern of $ \mathbf{v} $ can be equivalently expressed as $ f_{\rm SLA}( r, \theta;{\bar{\mathbf{v}}},M)=\left| \mathbf{b}^H_{\rm SLA}(r,\theta)\bar{\mathbf{v}} \right| $,
where  
\begin{equation}
	\bar{\mathbf{v}}^{T} \!\triangleq\! \left[ [\mathbf{w}]_{-\frac{N-1}{2}}\!,\![\mathbf{w}]_{M-\frac{N-1}{2}}\!,\![\mathbf{w}]_{2M-\frac{N-1}{2}}\!,\!\cdots\!,[\mathbf{w}]_{\frac{N-1}{2}}\right].
\end{equation}
\begin{figure}
	\centering
	\includegraphics[width=8cm]{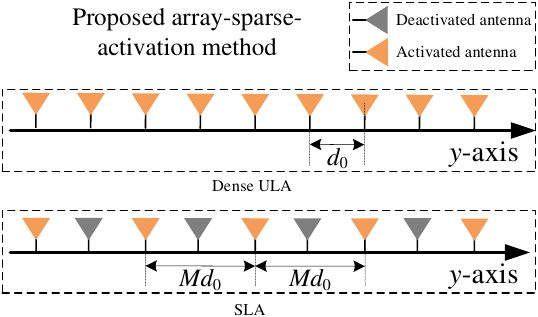}
	\caption{Illustration of proposed array-sparse-activation method.}
	\label{fig:Sparse activation difference}
	\vspace{-8pt}
\end{figure}Moreover, $ \mathbf{b}_{\rm SLA}^{H}\left(r,\theta\right) $ denotes the channel response vector of the effective SLA, which is given by
\begin{equation}
	\left[\mathbf{b}^H_{\rm SLA}\left(r,\theta\right) \right]_{q} = \frac{1}{\sqrt{Q}}e^{-\frac{\jmath 2 \pi r_{q}}{\lambda}}, \forall q\in \mathcal{Q}.
\end{equation}
Herein, $ \mathcal{Q} \triangleq \{ 0, \pm 1, \cdots, \pm \frac{Q-1}{2} \} $ denotes the set of the SLA antenna index and $ r_{q} = \sqrt{r^{2} -2 q  M d_0 r \theta + (q M d_0)^2} $ denotes the range from user $k$ to antenna $ q $ of the activated SLA.
Similarly to \eqref{eq:fresnelApproximation}, $ r_{q} $ can be approximated as $ r_{q} \approx r-qMd_0 \theta+\frac{q^2 (Md_0)^2 (1- \theta^2)}{2 r} $ under the Fresnel approximation.

\subsection{Beam Pattern of Proposed Array-sparse-activation Method}
\label{Sec:beampattern}
For the proposed array-sparse-activation method, its beam pattern is characterized as follows.
\begin{lemma}
	\emph{For the array-sparse-activation based beamforming vector $ \bar{\mathbf{v}} $ sampled from $ \mathbf{w} = \mathbf{b}(r_{0}, \theta_{0}) $, its beam pattern is given by 
		\begin{align}
			\label{LSAsum}
			&f_{\rm SLA} \left( r, \theta; \bar{\mathbf{v}},M \right) =|\mathbf{b}_{\rm SLA}^H(r,\theta) \mathbf{b}_{\rm SLA}(r_{0},\theta_{0})| \nn\\
			\!\!\! 	\overset{(a_1)}{\approx}& \frac{1}{Q} \left| \sum_{q\in \mathcal{Q}}\!\exp{\l(\jmath  \underbrace{{\pi} q  M \Delta_{\rm NF}}_{B_1}+\jmath \underbrace{\frac{\pi}{\lambda} q^2 ( M d_{0})^2\Phi }_{B_2 }\r)} \right|\\
			\triangleq & \hat{f}({r},{\theta};\bar{\mathbf{v}},M )\nn, ~\forall r, \theta,
			\!
		\end{align}
	\noindent where $ \Delta_{\rm NF}\triangleq \theta \!-\!\theta_0$ is defined as the \emph{spatial angle difference}, $ \Phi \triangleq\!  \frac{{1-\theta_{0}^2 }}{{r_{0}}} \!-\! {\frac{1-\theta^2 }{r}}$ is named as the \emph{ring difference} \cite{zhou2024sparse}, and $(a_1)$ is due to the Fresnel approximation, which has been shown to be accurate in \cite{kosasih2023finite}. 
}
\end{lemma}
To obtain useful insights, we further characterize the beam pattern of $ \hat{f}({r},{\theta};\bar{\mathbf{v}},M ) $ in the following two cases, depending on the value of $ B_2 $ in \eqref{LSAsum} (which will be shown to be determined by $ M $ equivalently).
\begin{itemize}
	\item {\bf Case 1}: $ B_2 \in \mathcal{U} \triangleq \{ B_2 = \frac{\alpha}{2}\pi, \forall q \in \mathcal{Q} \}$, where $ \alpha \in \mathcal{Z}\setminus \{0\}$ and $ \mathcal{Z} $ denotes the integer set.
	In this case, we have $ e^{\jmath B_2} = \pm 1 ~{\text{or}}~ \pm\jmath$.
	As such, it can be shown in Section \ref{sec:Desired multi-beam pattern} that the beam pattern $ \hat{f}({r},{\theta};\bar{\mathbf{v}},M ) $ in \eqref{LSAsum} will degenerate to the far-field scenario in the angular domain for multiple user ranges. This renders the same received signal power at multiple ranges  at the user angle, hence affecting the beam determination in the range domain.
	\item {\bf Case 2}: When $ B_2 \notin \mathcal{U} $, it is the desired multi-beam pattern.
\end{itemize}

\begin{figure}
	\centering
	%\vspace{-10pt}
	\includegraphics[width=8.5cm]{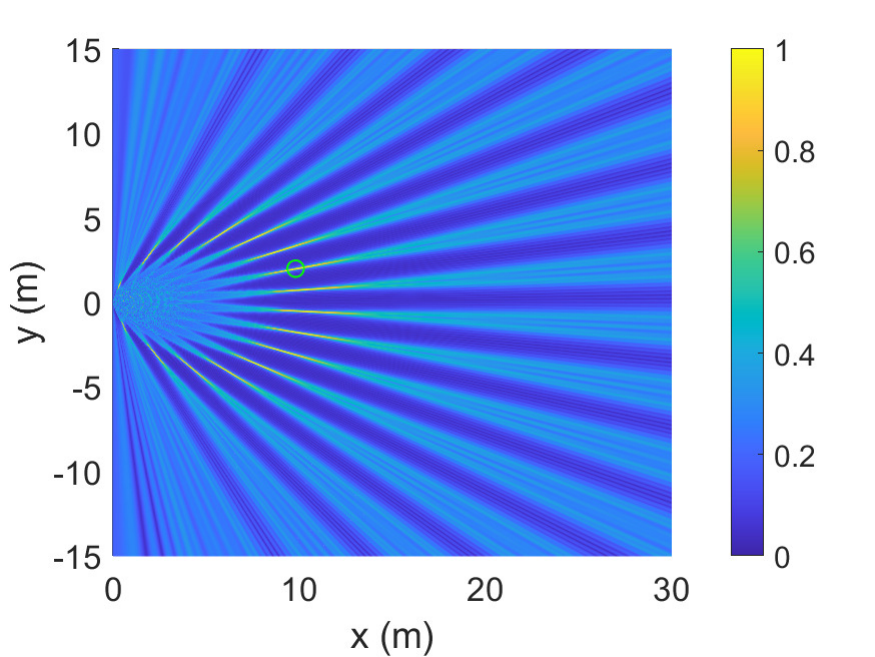}
	\caption{Near-field beam pattern of the beamforming vector $ \mathbf{v} $ sampling from $ \mathbf{w} $, where $N=257$, $ M = 16 $ and $f=30$ GHz. The codeword $ \mathbf{w} $ steers the beam towards $ \theta_0 = 0.2, r_0 =10 $ m which is marked by the green ring.}
	\label{fig:gratingLobe}
	\vspace{-5pt}
\end{figure}
\subsubsection{Desired multi-beam pattern}
\label{sec:Desired multi-beam pattern}
We first consider Case 2 for which $ B_2 \notin \mathcal{U} $, while the corresponding occurrence condition will be explicitly presented in Section \ref{sec:extraRings}.
For this case, we analytically show below that the array-sparse-activation method essentially provides a multi-beam pattern for near-field communications.
\begin{proposition}[Near-field multi-beam pattern]
	\label{theorem:multi-beam}
	{\rm Given $ B_2 \notin \mathcal{U}$, the beamforming vector $\bar{\mathbf{v}}$ based on array sparse activation generates a multi-beam pattern, with $ M $ beams focused on
	\begin{equation}
		\begin{aligned}
			\label{eq:gratingLobeAngle}
			\theta_{m} = \theta_0 + \frac{2m}{M},
			r_{m} = r_0\frac{1- \theta_m^2}{1-\theta_0^2},
			 \forall m \in \mathcal{M},
		\end{aligned}
	\end{equation}
	where $ \mathcal{M}  \triangleq  \{0, \pm 1, \!\dots,\! \pm ( M - 1)\} \cap \{m | -1 \le  \theta_0\!+\!\frac{2m}{M} \le 1 \} $.
	Specifically, all the sub-beams have the same beam width
	\begin{equation}
		\label{eq:beam-width}
		{\rm BW}_{\bar{\mathbf{v}}, m} = \frac{4}{QM}, \forall m \in \mathcal{M}.
	\end{equation} 
	Moreover, the beam-depths of the sub-beams are given by
	\begin{equation}
		\label{eq:beam-depth}
		\mathrm{BD}_{\bar{\mathbf{v}}, m}= \begin{cases}\frac{2 r_{m}^2 r_{\mathrm{BD},m}}{r_{\mathrm{BD},m}^2-r_{m}^2}, & r_{m}<r_{\mathrm{BD},m} \\ \infty, & r_{m} \geq r_{\mathrm{BD},m}\end{cases}, \forall m \in \mathcal{M},
	\end{equation}
	\rm where $ r_{\mathrm{BD},m} \approx \frac{
		Q^{2} M^{2}d_0 (1-\theta^{2}_{m})} {4  \varphi _{3 \mathrm{dB}}^2}$ and $ \eta _{3\rm dB} = 1.25 $ \cite{9723331}.
}
\end{proposition}
\begin{proof}
	We can divide the proof into three conditions.
	1) Condition 1: $ B_1 \neq  2\alpha \pi$ and $ B_2 = 0$ with $ \alpha \in \mathcal{Z}$ as $ q $ changes.
	In Condition 1, the user ring is obtained where we can obtain the beam-width and angles of multi-beams.
	2) Condition 2: $ B_1 = 2\alpha \pi$ and $ B_2 \neq 0$ with $ \alpha \in \mathcal{Z}$ as $ q $ changes.
	In Condition 2, we can obtain the beam-depth in $ M $ multiple beam angles.
	3) Condition 3: $ B_1 \neq 2\alpha \pi$ and $ B_2 \neq 0$ for arbitrary integer $ \alpha $ as $ q $ changes.
	We can show that only when $ \Delta_{NF} \to 0 $, the beam pattern has significant values.
	Based on the above, \eqref{eq:gratingLobeAngle}--\eqref{eq:beam-depth} can be obtained by following similar procedures in \cite{zhou2024sparse}, which are omitted due to limited space.
\end{proof}

According to Proposition \ref{theorem:multi-beam}, the beamforming vector based on array sparse activation generates $ M $ beams in the polar domain, which are focused on different locations of the user-ring.
Among the $ M $ beams, one is the main-lobe, while the others are grating-lobes induced by the SLA, which share the same beam-width and duplicate beam-depth in the corresponding angles \cite{zhou2024sparse}.

\begin{remark}[Multi-beam codebook]
	\label{remark:periodicity in the agular domain}
	{\rm Proposition \ref{theorem:multi-beam} indicates that in the near-field, a beamforming vector $ \bar{\mathbf{v}} $ with the array activation parameter $ M $ generates $ M $ beams with a }  {\rm period of $ {2}/{M} $ in the angular domain, as illustrated in Fig. \ref{fig:gratingLobe}. Therefore, we can divide the entire spatial angular domain $ [-1,1] $ into $ M $ uniform sectors. Then, we only need to steer single-angle beams in the angular subspace} $ [-\frac{1}{M},\frac{1}{M}) $ {\rm (a period)} {\rm , for which $ (M-1) $ beams (grating-lobs) will appear in the other subspaces.
}
\end{remark}

In the next subsection, we consider the other case $B_2 \in \mathcal{U}$, where several abnormal rings appear which may affect the performance of multi-beam training method to be designed.
\subsubsection{Undesired multi-ring beam pattern}
\label{sec:extraRings}
When $ B_2 \in \mathcal{U} $, we have $ e^{\jmath B_2} = \pm 1 ~{\text{or}}~ \pm\jmath$. In this case, $ \Phi = \frac{2\alpha}{M^{2}d_{0}}$, which equivalently enforces
\begin{equation}
	\frac{1- \theta^2}{r} = \frac{-\alpha}{M^{2}d_{0}} + \frac{1-\theta_0^2}{r_0}, \forall \alpha.
\end{equation}
For ease of exposition, we assume that $ (Q-1)/2 $ is an even number.
Moreover, when $ \alpha > 0 $ and $ \theta = \theta_0 $, we always have $ r > r_0 $, which means that the abnormal rings are outside the user ring, and vice versa. 
Therefore, we call the abnormal rings as outer abnormal rings with $ \alpha > 0 $ and inner abnormal rings when $ \alpha < 0 $.
In particular, the abnormal rings can be classified into three categories, simply named as the {\emph{ Type-I, Type-II and Type-III rings}} depending on the value of $ \alpha $.
In the following, we characterize the beam pattern of the abnormal rings as follows.
\begin{lemma}
	\label{lemma:Properties of Type II extra rings}
	\emph{Given $ \frac{1- \theta^2}{r} = \frac{-\alpha}{M^{2}d_{0}} + \frac{1-\theta_0^2}{r_0}, \forall \alpha \in \mathcal{Z} $, the beam pattern of codeword $\bar{\mathbf{v}}$, i.e., $ \hat{f}({r},{\theta};\bar{\mathbf{v}},M ) $ in \eqref{LSAsum}, reduces to
		{\small 
			\begin{equation}
				\label{eq:extra-ring beam pattern}
				\hat{f}({r},{\theta};{\bar{\mathbf{v}}},M)\!\!=\!\!
				\left\{\begin{matrix}
					\!\!\!\!\!\!\!\!\!\!\!\!\!\!\!\!\!\!\!\!\!\!\!\!\!\!\!\!\!\!\!\!\!\!\!\!\!\!\!\!\!\!\!\!\!\!\!\!\!\!\!\!\!\!\!\!\frac{\left|\Xi_{Q}(M \Delta_{\rm NF} )\right|}{Q}, {\text{when}}~\alpha \!\!=\!\! 4k_1, 
					\\ \!\!\!\frac{\left|\Xi_{\frac{Q+1}{2}}(2M\Delta_{\rm NF}) \!-\! \Xi_{\frac{Q-1}{2}}(2M\Delta_{\rm NF}) \right|}{Q}, {\text{when}}~\alpha \!\!=\!\! 4k_2 \!\!+ \!\!2,
					\\ \!\frac{\left|\Xi_{\frac{Q+1}{2}}(2M\Delta_{\rm NF}) \!+\! \jmath\Xi_{\frac{Q-1}{2}}(2M\Delta_{\rm NF}) \right|}{Q}, {\text{when}}~\alpha \!\!=\!\! 4k_3\!\!+\!\!1,\\
					\!\frac{\left|\Xi_{\frac{Q+1}{2}}(2M\Delta_{\rm NF}) \!-\! \jmath\Xi_{\frac{Q-1}{2}}(2M\Delta_{\rm NF}) \right|}{Q}, {\text{when}}~ \alpha \!\!=\!\! 4k_4\!\!+\!\! 3,\\
				\end{matrix}\right.		
	\end{equation}}where $ k_1 \in \mathcal{Z} \setminus \{0\} $, $ k_2, k_3$ and $ k_4 \in \mathcal{Z}$. Moreover, $ \Xi_{\alpha}(x) \triangleq \frac{\sin(\frac{\alpha x \pi}{2})}{\sin(\frac{x \pi}{2})}$  is defined as the Dirichlet Sinc function.}
\end{lemma}
\begin{proof}
	For $ \alpha = 4k_1 $, we have $ \exp(\jmath B_2) =1 $ for arbitrary $ q $.
	Hence, the beam pattern is given by
	{\small
		\begin{equation*}
			\begin{aligned}
				\sum\limits_{q \in \mathcal{Q}}\!\exp(\jmath B_1) &\!\!=\!\! \frac{\exp(\jmath M\Delta_{\rm NF}\frac{-Q+1}{2})\left(1-\exp(\jmath QM\Delta_{\rm NF})\right)}{1-\exp(\jmath M\Delta_{\rm NF})}\\
				&=\Xi_{Q}(2M\Delta_{\rm NF})
			\end{aligned}
	\end{equation*}}For $ \alpha = 4k_1 $, we have $ \hat{f}({r},{\theta};\bar{\mathbf{v}},M) = \frac{\Xi_{Q}(2M\Delta_{\rm NF})}{Q}$.  For $ \alpha = 4k_2 + 2 $, when $ q $ is an even (or odd) number, we always have $ \exp(\jmath B_2) =1 $ (or $ \exp(\jmath B_2) = -1 $).
	Hence, we have {\small$\hat{f}({r},{\theta};\bar{\mathbf{v}},M) = \frac{1}{Q}\left| \sum\limits_{q\in \{\mathcal{Q} \cap  \mathcal{E}\}}\exp(\jmath B_1) - \sum\limits_{q\in \{\mathcal{Q} \cap  \mathcal{J}\}}\exp(\jmath B_1) \right|$}, where  $ \mathcal{J} $ denotes the odd integer set. Then, it can be shown 
{\small
	\begin{equation*}
	\begin{aligned}
	\sum\limits_{q\in \{\mathcal{Q} \cap  \mathcal{E}\}}\exp(\jmath B_1) =\Xi_{\frac{Q+1}{2}}(2M\Delta_{\rm NF}).
	\end{aligned}
	\end{equation*}}Similarly, we have $ \sum\limits_{q\in \{\mathcal{Q} \cap \mathcal{J}\}}\!\!\!\exp(\jmath B_1) =  \Xi_{\frac{Q-1}{2}}(2M\Delta_{\rm NF})$.
	Hence, when $ \alpha = 4k_2 + 2 $, we have {\small $ \hat{f}({r},{\theta};\bar{\mathbf{v}},M) = \frac{1}{Q}\left|\Xi_{\frac{Q+1}{2}}(2M\Delta_{\rm NF}) \!\!-\!\! \Xi_{\frac{Q-1}{2}}(2M\Delta_{\rm NF}) \right| $}.
	Similar to the proof for the case  $ \alpha = 4k_2 + 2 $, we can obtain the results for $ \alpha = 4k_2 + 1 $ and $ \alpha = 4k_4 + 3 $ as shown in \eqref{eq:extra-ring beam pattern}. 
\end{proof}
Next, we give the properties of the above abnormal rings and their occurrence conditions.
\begin{figure}
	\centering
	\includegraphics[width=8.6cm]{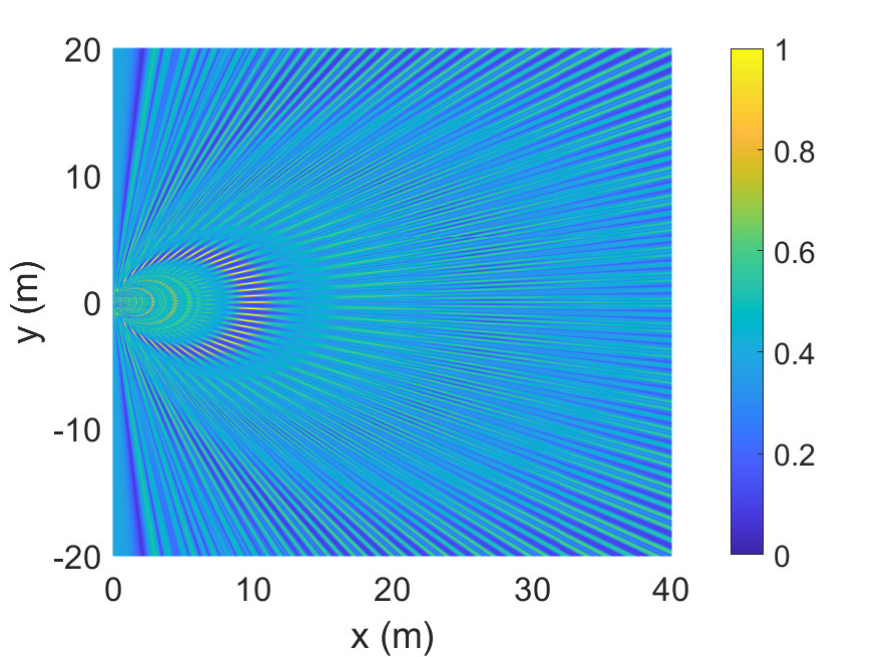}
	\caption{Near-field beam pattern with Type-III rings when $N=321$, $ M = 40 $ and $f=30$ GHz. The codeword steers the beam towards $ \theta_0 = 0, r_0 =10  $ {\rm m}.}
	\label{fig:differentRings}
	\vspace{-5pt}
\end{figure}
\begin{itemize}
	\item \textbf{Type-I ring}: For Type-I rings, we have  $ \alpha = 4k_1 $.
	The beam pattern in \eqref{eq:extra-ring beam pattern} indicates that when $ \theta = \theta_0 + \frac{2\gamma}{M}, \gamma = \pm 1, \cdots, \pm (M-1) $, there are $ M $ undesired beams on Type-I rings with $ \hat{f}({r},{\theta};\bar{\mathbf{v}},M ) = 1 $.
	\item \textbf{{{Type-II ring}}}:
	For Type-II rings, we have $ \alpha = 4k_2 + 2 $ and thus $ \hat{f}_{}({r},{\theta};\bar{\mathbf{v}},M) = 1 $ when $ \theta = \theta_0 + \frac{2\gamma+1}{M}, \gamma = \pm 1, \cdots, \pm (M-1) $.
	Moreover, we have $ \hat{f}_{}({r},{\theta};\bar{\mathbf{v}},M) = \frac{1}{Q} $, when $ \theta = \theta_0 + \frac{2\gamma}{M}, \gamma = \pm 1, \cdots, \pm (M-1) $
	\item \textbf{{Type-III ring}}:
	For Type-III rings, we have $ \alpha = 4k_3 + 1 $ or $ \alpha = 4k_4 + 3 $. Moreover, we can show that $ \hat{f}_{}({r},{\theta};\bar{\mathbf{v}},M) = \frac{|\frac{Q+1}{2}+\frac{\jmath(Q-1)}{2}|}{Q} \approx \frac{\sqrt{2}}{2} $ when $ \theta = \theta_0 + \frac{\gamma}{M}, \gamma = \pm 1, \cdots, \pm (2M-1) $ as illustrated in Fig. \ref{fig:differentRings}.
\end{itemize}

It is worth noting that for Type-I and Type-II rings, there exist multiple ranges that share the same beam power at the user angle, thus affecting the user range determination.
On the other hand, Type-III affects the polar-domain codebook design when the range sampling method in \cite{Cui2022channel} is applied, for which the coherence between adjacent codewords in the same angle is no larger than 1/2.
Although we can resolve the issue caused by Type-III rings by increasing the number of range samples, it may increase the beam training overhead.
Therefore, when designing the antenna activation interval $ M $, we should avoid occurrence of all these three types of abnormal rings.

The necessary conditions for the occurrence of the above-mentioned abnormal rings are given below.
\begin{proposition}
	\label{pro:peak-rings}
	{\rm Given a codeword $ \bar{\mathbf{v}} $ parameterized by} $ M $, {\rm if} $ M > M_{\rm th} \triangleq { \sqrt{1.2(N-1)}  } ${\rm, the inner and outer abnormal rings will simultaneously appear in the Fresnel region.}
	%	{\rm Moreover, when} $ M > \sqrt{\frac{4r_u}{d_0\cos^2\theta_u}} $, {\rm the outer peak-ring} \rm { will appear.}
\end{proposition}
\begin{proof}
	We first consider $ \alpha < 0 $ and have
	{$ r_{\rm max} =  {1}/\left({\frac{-\alpha}{M^{2}d_{0}} + \frac{1-\theta_0^2}{r_0}}\right)< \frac{M^{2}d_{0}}{-\alpha} \le {M^{2}d_{0}}.\vspace{-0.1cm}$} Hence, if the first inner abnormal ring exceeds Fresnel region, we have
	{$$ {M^{2}d_{0}} > 1.2D, $$}which satisfy {$ M > \sqrt{1.2(N-1)} $}. For $ \alpha > 0 $, when the outer abnormal rings appear, $ M $ should satisfy
	$\frac{-1}{M^{2}d_{0}} + \frac{1-\theta_0^2}{r_0} > 0. $
	Then, we have {$ M > \sqrt{\frac{r_0}{d_0 (1-\theta_0^2)}} > \sqrt{\frac{r_0}{d_0}} $}.
	For $ r_0 \ge 1.2D$, we have $ M > \sqrt{1.2(N-1)}$, thus completing the proof.
\end{proof}

Proposition \ref{pro:peak-rings} shows that when $ M $ is sufficiently large (i.e., $ M >M_{\rm th}$), there may appear undesired beam lobes and hence destroying the near-field multi-beam pattern. Nevertheless, it is worth noting that $ M_{\rm th} $ is practically large for XL-array systems.
For example, for the system setup with $ N = 513 $, the necessary condition for the occurrence of the abnormal rings is $ M_{\rm th} $ = 25, which is very large.

\vspace{-0.2cm}
\subsection{Multi-beam Codebook Design}
\label{multi-beam codebook design}
In this subsection, we design the new near-field multi-beam codebook.
Specifically, based on the multi-beam pattern generated by the proposed array-sparse-activation method, we only need to design the codebook in the angular subspace $ [-\frac{1}{M},\frac{1}{M})$. 
%Similar to \cite{Cui2022channel}, we set $ \Phi = 0 $ and $ \theta = \theta_0 $ to respectively obtain angular and distance sampling methods.
%Then, we give the sampling methods in the angular and range domain.
To this end, the angular and range domains are sampled as follows, by using similar methods in \cite{Cui2022channel}.

\subsubsection{Angle sampling method}
Similar to \cite{Cui2022channel}, we set $ \Phi = 0 $ (called  user-ring \cite{cui2022near}) to obtain the angular sampling method. Then, the near-field beam pattern in \eqref{LSAsum} is obtained as below.
\begin{lemma}\emph{
	\label{lemma:LSA_AngularDomain}
	When $\Phi=0$ (or equivalently $ \frac{1-\theta^2}{r} =  \frac{ 1-\theta_0^2}{r_0}$), the beam pattern $\hat{f}({r},{\theta};\bar{\mathbf{v}},M) $  in \eqref{LSAsum} is given by
	\begin{equation}
	\label{eq:Sinc}
	\begin{aligned}
	\hat{f}({r},{\theta};\bar{\mathbf{v}},M) = \frac{1}{Q}
	\left| \Xi_{Q}(M \Delta_{\rm NF} ) \right|.
	\end{aligned}		
	\end{equation}}
\end{lemma}

Obviously, when $ \frac{\pi QM\Delta_{\rm NF}}{2} = \gamma\pi$, i.e., $ \Delta_{\rm NF} =  \frac{2\gamma}{QM}$ with an arbitrary integer $ \gamma $, $ \hat{f}_{\rm NF}({r},{\theta};\bar{\mathbf{v}},M) $ equals to zero.
Therefore, in the angular domain, we should sample the angle with an interval of $ \frac{2}{QM} $ on the user-ring \cite{Cui2022channel}. Then, the sampled angles are given by
\begin{equation}
	\label{eq:angleSampling}
	\theta_s = \frac{2s - QM - 1}{QM}, \forall s = 1,\cdots,QM.
\end{equation}
 For the angular resolution, it is noted that the activated SLA for the codeword $ \bar{\mathbf{v}} $ has the same aperture with the original ULA, i.e., $ (Q - 1)Md_0 = (N-1)d_0 ${\rm. As such, we have} 
	\begin{equation}
	\label{eq:AngleBig}
		QM > N
		~~
		({\text{or equivalently}}~~ \frac{2}{QM} < \frac{2}{N}).
	\end{equation}
Combining \eqref{eq:angleSampling} and \eqref{eq:AngleBig}, it can be obtained that the angle sampling interval of the activated SLA is smaller than that of the ULA, thereby achieving a higher angular resolution for the beam training codebook.

\subsubsection{Range sampling method}
For the range domain, we set $ \theta = \theta_0 $ and design the range sampling method, similar as in \cite{Cui2022channel}. To this end, the beam pattern for \eqref{LSAsum} is approximated as follows \cite{zhou2024sparse}.
\begin{lemma}
	\label{lemma:LSA_DistanceDomain}
	\emph{When $ \theta = \theta_{0} $, the beam pattern in \eqref{LSAsum} can be approximated as \cite{zhou2024sparse}
		\begin{equation}
		{\hat{f}(r,\theta;\bar{\mathbf{v}},M) \approx}
		\left| F(\gamma)\right|,
		\end{equation}
		where $F(x)  \triangleq \frac{C(x)+\jmath S(x)}{x}$ with $\gamma = \frac{Q M}{2}\sqrt{{d_0}\left| \Phi \right|}$.
	}
\end{lemma}
To ensure the power loss no larger than $ 3 $ dB, we set $ F(\gamma) \ge \frac{1}{2} $, which is equivalent to  \cite{liu2023near} $$ \gamma = \frac{Q M}{2}\sqrt{{d_0}\left| \frac{1-\theta_0^2}{r_0} - \frac{1-\theta_0^2}{r} \right|} \le  \varphi _{3\rm dB} \triangleq 1.6. $$ Then, we have
\begin{equation}
	\left| \frac{1}{r_0} - \frac{1}{r} \right| \le \frac{2 \lambda \varphi_{3\rm dB}^2}{M^2Q^2 d_0^2(1-\theta_0^2)} \triangleq \frac{1}{Z_{\rm LSA}(1-\theta_0^2)},
\end{equation}
which leads to the sampling grids $ r_v = \frac{1}{v}Z_{\rm LSA}(1-\theta_s^2) , \forall v \in \mathcal{V} \triangleq \{1,2,\cdots, V\} $, where $ V $ denotes the number of range samples\cite{Cui2022channel}.

\subsubsection{Codebook design}
According to the above angle and range sampling methods as well as the multi-beam pattern in Proposition \ref{theorem:multi-beam}, we design the array-sparse-activation based near-field multi-beam codebook as follows. Specifically, we only consider the polar-domain subspace within the angular interval $ [-\frac{1}{M},\frac{1}{M}) $. The multi-beam codebook, denoted by $ {\bar{\mathcal{V}}}_{\rm{Pol}} $, consists of $ QV $ codewords, each steering its main-lob towards the location ($ r_{g, v}, \theta_g $), where $ g \in \mathcal{G}\triangleq\{\frac{Q(M-1)+1}{2},\frac{Q(M-1)+1}{2}+1,\cdots, \frac{Q(M+1)-1}{2}\}  $ and $v \in \mathcal{V}$. Mathematically, we have
\begin{framed}
	{\setlength\abovedisplayskip{8pt}
		\setlength\belowdisplayskip{0pt}
		{\bf Near-field multi-beam codebook:}
		\begin{equation} 
		\label{Eq:LSApolarDomainCodebook}
		{\bar{\mathcal{V}}}_{\rm{Pol}}\!\! =\!\! \{{\bar{\mathcal{V}}}_{\frac{Q(M-1)+1}{2}},\!\cdots\!,{\bar{\mathcal{V}}}_g,\cdots,{\bar{\mathcal{V}}}_{\frac{Q(M+1)-1}{2}} \},
		\end{equation} 
}\end{framed}\noindent where ${\bar{\mathcal{V}}}_g\triangleq\{{\bar{\mathbf{v}}}_{g,1}, \cdots,{\bar{\mathbf{v}}}_{g,v},\cdots {\bar{\mathbf{v}}}_{g,V}\}$.
Moreover, we have $ {\bar{\mathbf{v}}}_{g,v} =  \mathbf{b}_{\rm SLA}\left(r_{g, v}, \theta_g\right)$ with  $ r_{g, v} = \frac{1}{v}Z_{\rm SLA}(1-\theta_g^2) $.
Note that by using the array-sparse-activation method, the other subspaces can be covered by the grating-lobes of the designed multi-beam codebook. Let $ {\mathcal{W}}_{\rm{Pol}} $ denote the single-beam codebook in the angular domain consisting of $ QMV $ codewords, which is given by 
\begin{equation}
	\label{Eq:ULApolarDomainCodebook}
	{\mathcal{W}}_{\rm{Pol}} = \{{\mathcal{W}}_1,\cdots,{\mathcal{W}}_s,\cdots,{\mathcal{W}}_{QM} \},
\end{equation}
%\end{center}
\noindent where ${\mathcal{W}}_s\triangleq\{{\mathbf{w}}_{s,1}, \cdots,{\mathbf{w}}_{s,v},\cdots {\mathbf{w}}_{s,V}\}$ and $ {\mathbf{w}}_{s,v} =  \mathbf{b}\left(r_{s, v},\theta_s\right)$ with $ r_{s, v} = \frac{1}{v}Z_{\rm SLA}(1-\theta_s^2) $.
Then we have the following key result.

\begin{proposition}
	\label{proposition:multiple codeword}
	{\rm For the proposed near-field multi-beam codebook in \eqref{Eq:LSApolarDomainCodebook},  the beam coverage of each codeword $ {\bar{\mathbf{v}}}_{s,v} $ is the same as that of $ M $ codewords $ {{\mathbf{w}}}_{s+mQ,v}, \forall m \in \mathcal{M} $ in the single-beam codebook \eqref{Eq:ULApolarDomainCodebook}.	 
} 
\end{proposition}
\begin{proof}
	For each codeword $ {\bar{\mathbf{v}}}_{s,v} = \mathbf{b}_{\rm SLA}\left(r_{s, v}, \theta_s\right) $ that steers multiple beams towards $ \theta_{m} = \theta_s + \frac{2m}{M},
	r_{m} = r_{s,v}\frac{1- \theta_m^2}{1-\theta_s^2},
	\forall m \in \mathcal{M} $, we have$
 \theta_{m} = \theta_s + mQ\times\frac{2}{QM} = \theta_{s+mQ}$,
	$r_{m} = \frac{1}{v}Z_{\rm SLA}(1-\theta_s^2)\frac{1- \theta_{s+mQ}^2}{1-\theta_s^2} = \frac{1}{v}Z_{\rm SLA}(1-\theta_{s+mQ}^2) = r_{s+mQ},
	\forall m \in \mathcal{M} $.
	Therefore, the beam coverage of codeword $ {\bar{\mathbf{v}}}_{s,v} $ is the same as that of  $ {{\mathbf{w}}}_{s+mQ,v} = \mathbf{b}\left(r_{s+mQ, v}, \theta_{s+mQ}\right) $,  thus completing the proof.
\end{proof}
\section{Proposed Multi-beam Training Method}
\label{Sec:scheme1}
In this section, we propose an efficient near-field multi-beam training method based on the designed multi-beam codebook.
For ease of exposition, we focus on the beam training for a typical user and thus omit the user index in the sequel.

The proposed method consists of two phases, namely, the \emph{multi-beam based candidate user locations estimation} and the \emph{single-beam based user location identification}.
The key idea is to first determine $ M $ candidate user locations by using the designed  multi-beam codebook and then resolve the user location ambiguity by sequentially testing $ M $ single-beam codewords. The detailed procedures are presented below.
\subsubsection{Multi-beam based candidate user locations estimation}
\label{sec:Multi-beam sweeping}
In this phase, the XL-array BS utilizes the multi-beam codebook $ {\bar{\mathcal{V}}}_{\rm{Pol}} $ in \eqref{Eq:LSApolarDomainCodebook} to scan the whole space and identify $ M $ candidate locations based on the received signals. 
Specifically, the BS sequentially sends $ QV $ training symbols by employing the multi-beam codebook $ {\bar{\mathcal{V}}}_{\rm{Pol}} $, as illustrated in Fig. \ref{fig:trainingMethod}(a). 
For each codeword $ {\bar{\mathbf{v}}}_{g,v} $, the received signal at the user is given by
\begin{align}\label{Eq:received signal of each codeword}
	y_{}({\bar{\mathbf{v}}}_{g,v}) =
	%    \mathbf{h}^H_{\rm near}\mathbf{v}x+z
	\sqrt{N}\beta\mathbf{b}^{H}(r_{0},\theta_{0}){\bar{\mathbf{v}}}_{g,v}x+z.
\end{align}
Next, the user estimates the best multi-beam codeword based on the received powers, whose index can be mathematically expressed as 
\begin{equation}
	\left(\check{g}, \check{v}\right)=\arg \max _{g \in \mathcal{G}, v \in \mathcal{V}}\left|y_{}({\bar{\mathbf{v}}}_{g,v})\right|^2 .
\end{equation}
Then, the user feeds back the index $ \{\check{g},\check{v}\} $ to the BS via an error-free control link.
After receiving the feedback, the BS determines $ M  $ candidate user locations, given by
\begin{equation}
	\theta_{m} = \theta_{\check{g}} + \frac{2m}{M},~~~ r_{m} = r_{{\check{g}},{\check{v}}}\frac{1- \theta_{m}^2}{1-\theta_{\bar g}^2}, \forall m \in \mathcal{M}.
\end{equation}
\begin{figure}[t]
	\includegraphics[width=9cm]{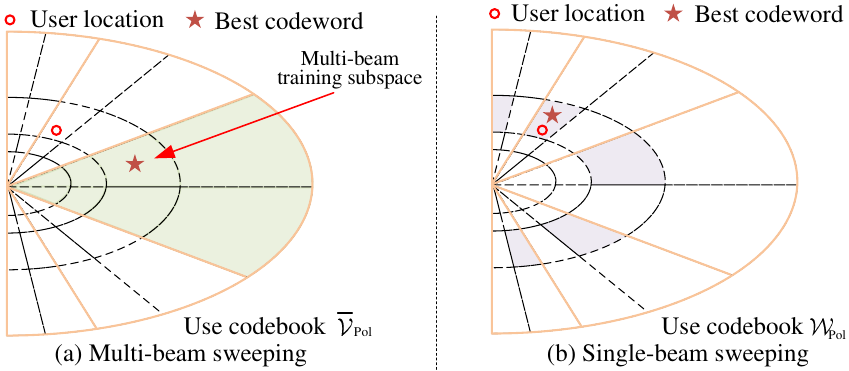}
	\centering
	\caption{Illustration of the proposed multi-beam training method based on array sparse activation.}
	\label{fig:trainingMethod}
	%\vspace{-10pt}
\end{figure}
\subsubsection{Single-beam based user location identification}
\label{sec:Single-beam sweeping}
In the second phase, the BS further resolves the user location ambiguity by employing the single-beam codebook for beam sweeping.
Specifically, the BS sequentially sends $ M $ training symbols by using the beam codewords $ \{ {\mathbf{w}}_{\check{g}+mQ,\check{v}},\forall m \in \mathcal{M} \}  $ as illustrated in Fig. \ref{fig:trainingMethod}(b).
For each codeword $ {\mathbf{w}}_{\check{g}+mQ,\check{v}} $, the received signal at the user is given by
\begin{align}\label{Eq:received signal single beam}
	y_{}({\mathbf{w}}_{\check{g}+mQ,\check{v}}) \!=\!
	\sqrt{N}\beta_{}\mathbf{b}^{H}(r_{0},\theta_{0}){\mathbf{w}}_{\check{g}+mQ,\check{v}}x+z.
\end{align}
Then, based on the received powers, the user identifies the best single-beam polar-domain codeword, whose index is given by
\begin{equation}
	(g^\ast,{v}^\ast) =\arg \max _{m \in \mathcal{M}}\left|y_{}({\mathbf{w}}_{\check{g}+mQ,\check{v}})\right|^2 .
\end{equation}
Accordingly, the user location is estimated as $ (r^\ast,\theta^\ast) $ based on the polar-domain codebook $ {\mathcal{W}}_{\rm{Pol}} $.

{\underline{\textbf{Beam training overhead:}}}
The beam training overhead of the proposed method is analyzed as follows. For a typical user, the first phase requires $ QV $ training symbols for the multi-beam sweeping in the sector subspace (see Section \ref{sec:Multi-beam sweeping}) and the second phase requires $ M $ training symbols for resolving the angular ambiguity (see Section \ref{sec:Single-beam sweeping}). Therefore, for $ K $ users, in the worst case, the total number of  beam training symbols is obtained as 
\begin{equation}
	\label{Eq:multi-beam training overhead}
	T^{(\rm MU)} = QV + KM,
\end{equation}
where $ Q = \frac{N-1}{M}$.
One can observe from \eqref{Eq:multi-beam training overhead} that as $ M $ increases, the beam training overhead of the first phase decreases, while that of the second phase increases due to more candidate user angles.
Thus, there generally exists a fundamental trade-off between the reduction of the training overhead in the first and second phases with respect to the value of $ M $.
As such, it is necessary to find an optimal $ M $ to minimize the total number of training symbols required for the proposed multi-beam training method, which can be mathematically formulated as follows.
\begin{subequations}
	\begin{align}
		({\bf P3}):~~~& \min_M~ F(M) = \frac{(N-1)V}{M} + KM \nn\\
		&~~{\text{s}}{\text{.t}}{\rm{. }} ~ M \le M_{\rm th},\label{Eq:desired beam pattern condition}\\
		&~~~~~~~ M \in \mathcal{F} \triangleq \left\{\frac{N-1}{M}+1\in \mathcal{Z}\right\} \label{Eq:interger constraint},
	\end{align}
\end{subequations}where $ M_{\rm th} $ is given in Proposition~\ref{pro:peak-rings} and constraint \eqref{Eq:desired beam pattern condition} is to enforce the desired multi-beam pattern (see details in Section \ref{sec:extraRings}). Problem ({\bf P3}) is a non-convex combinatorial optimization problem due to the integer constraint in \eqref{Eq:interger constraint}. To address this issue, we first relax the integer constraint for $ M $, thus yielding the following relaxed problem. 
\begin{subequations}
	\begin{align}
	({\bf P4}):~~~& \min_M~ F(M) = \frac{(N-1)V}{M} + KM\\
	&~~{\text{s}}{\text{.t}}{\rm{. }} ~ M \le M_{\rm th}.
	\end{align}
\end{subequations}
Problem ({\bf P4}) now is a convex optimization problem, whose optimal solution is obtained below.
\begin{lemma}
	\emph{The optimal solution to $ ({\bf P4}) $ is given by
			{\small\begin{equation}
				\label{eq:optimal solution of P3}
				M^{\ast} = \min\left\{M_{\rm th}, \sqrt{\frac{(N-1)V}{K}} \right \}. 
			\end{equation}}}
\end{lemma}

When the integer constraint in $ \eqref{Eq:interger constraint} $ is taken into account, a suboptimal solution to $ ({\bf P3}) $ can be obtained as
\begin{equation}
	\hat M = \arg \min\limits_{f \in  \mathcal{F}}\{  |M^{\ast} - f| \}.
\end{equation}

\begin{example}[Beam training overhead comparison] \emph{Consider the case with $ N = 257 $ , $ V = 4 $ and $ K = 8 $.
It can be shown that $ M = 16 $ is the optimal solution to Problems ({\bf P3}) and ({\bf P4}). As such,
the proposed multi-beam training method requires $ T^{(\rm MU)} = 256 $ symbols in total, which is significantly smaller than that of the exhaustive-search method ($ T^{(\rm EX)} = 4224 $) \cite{Cui2022channel} and the two-phase beam training method ($ T^{(\rm TP)} = 720 $)~\cite{Zhang2022fast}.}
\end{example}

\begin{remark}[Estimation error] \label{remark:estimation_error}
	\emph{Consider the case without noise. It can be easily shown that the proposed multi-beam training method can achieve the same beam training accuracy as the exhaustive-search method.
	Next, for the case with noise,
	the accuracy of the proposed multi-beam training method is mainly affected by two factors.
	First, the proposed method is an on-grid user angle-and-range estimation method. 
	A higher sampling resolution usually leads to higher estimation accuracy.
	Second, the beam training accuracy typically depends on the sparse-activation parameter $ M $. As $ M $ increases, the presence of abnormal rings may degrade the beam training accuracy, which will be numerically verified in Section \ref{Sec:numericalResults}. 
}
\end{remark}

\begin{remark}[Far-field extension]
\label{re:Far-field extension-scheme1}
{	\rm
	For the far-field scenario, the proposed multi-beam training method can be directly applied, while its beam training overhead can be further reduced if the polar-domain codebook is replaced by the DFT codebook.
	Specifically, the multi-beam codebook based on array sparse activation can be constructed as $ {\bar{\mathcal{V}}}_{\rm{DFT}}\!\! =\!\! \{{\bar{\mathbf{v}}}_{\frac{Q(M-1)+1}{2}},\!\cdots\!,{\bar{\mathbf{v}}}_g,\cdots,{\bar{\mathbf{v}}}_{\frac{Q(M+1)-1}{2}} \} $, where each codeword $ {\bar{\mathbf{v}}}_g $ steers the main-lobe towards the angle $ \theta_g $, and at the same time, generates $ M-1 $ grating-lobs in the angles $ \theta_g + 2/M, \forall m \in \mathcal{M} $.
	Then, in the first phase, we can employ the DFT-based multi-beam codebook for the angle sweeping in the sector subspace $ [-1/M, 1/M) $, which aims to determine $ M $ candidate user angles.
	In the second phase, the user angular ambiguity can be resolved by performing the single-beam sweeping based on the DFT codebook.
	As such, in the worst case, the total number of training symbols required is given by $ T^{(\rm MU)}_{\rm far} = Q+KM $.
}
\end{remark}
\begin{remark}[Extension to multi-path channels]
	{\rm 
	The proposed near-field multi-beam training method can be extended to the general case with comparable multi-path components. Specifically, in the first phase, each user needs to feed back the indices of multi-beam codewords, for which the received powers are larger than a properly chosen threshold.
	Then, in the second phase, the best angle and range for each path can be estimated by using the single-beam codebook, whose procedures are  similar to that for the LoS path.}
\end{remark}

\section{Numerical Results}
\label{Sec:numericalResults}
In this section, numerical results are presented to demonstrate the effectiveness of our proposed near-field multi-beam training method.
\subsection{System Setup and Benchmark Schemes}
The system parameters are set as follows, unless specified otherwise.
The BS is equipped with $ N = 257 $ antennas and the carrier frequency is $ f = 30 $ GHz, for which the reference channel power gain is  $ \beta_0 = (\frac{4\pi}{\lambda})^2  = -62 $ dB. Moreover, the BS transmit power and noise power are set as $ P_{\rm tol} = 30 $ dBm and $ \sigma^2 = -70 $ dBm, respectively.
The number of NLoS paths and Rician factor are respectively $ L_k = 2$ and $ \kappa_k = 30$ dB, $ \forall k $, which is practical in mmWave frequency bands \cite{liu2023low, shi2023spatial}.
The antenna activation interval is $ M = 16 $.
The reference SNR is defined as $ {\rm SNR} = \frac{NP_{\rm tol} \beta_0}{r_{\rm u}^2 \sigma^2} $~\cite{Zhang2022fast}, where $ r_{\rm u} $ denotes the user range.
In addition, the \emph{effective} sum-rate is defined as $ R_{\mathrm{e}}=\left(1-\frac{T_{\mathrm{p}}}{T_{\mathrm{tol}}}\right) R_{\text {sum }} $, where $ T_{\mathrm{p} }$ and $ T_{\mathrm{tol}} $ denote the time duration of required pilot symbols and the total transmission time, respectively.
In this paper, we set $ T_{\mathrm{tol}} = 0.2 $ ms and the time for transmitting a pilot symbol is 0.1 \si{\micro\second} \cite{liu2024near}.
All numerical results are obtained by averaging 1000 channel realizations.
For performance comparison, we consider the following benchmark schemes:

\begin{itemize}
	\item {\rm \textbf{Perfect-CSI based beamforming:}}
	This scheme assumes that the user channel information is perfectly known at the BS, which serves as the performance upper bound.
	\item {\rm \textbf{LS channel estimation:}}
	For this scheme, each user estimates the channel by using the LS downlink channel estimation method based on orthogonal pilot symbols \cite{coleri2002channel}.
	The overhead of this scheme is $ T^{({\rm LS})} = N $. 
	\item {\rm \textbf{Exhaustive-search beam training:}} 
	This scheme utilizes the single-beam near-field codebook (see \eqref{Eq:ULApolarDomainCodebook}) to conduct an exhaustive search over both the angular and range domains \cite{Cui2022channel}.
	Its beam training overhead is $ T^{\rm (EX)} = QMV $.
	\item {\rm \textbf{Two-phase beam training:}}
	This scheme first utilizes the DFT codebook to obtain the user angle information  and then estimates the user range by employing the polar domain codebook.
	To ensure superior performance, we assume that $ \chi = 3 $ candidate angles are selected in the first phase \cite{Zhang2022fast}, for which the total beam training overhead is $ T^{\rm (2P)} = QM + 3KV $.
	\item {\rm \textbf{Multi-beam training based on sub-array:}}
	This far-field based beam training scheme leverages the array division method to simultaneously generate multiple beams.
	However, this scheme suffers from significant performance loss when the number of sub-array  $ \tilde M $ is large \cite{you2022fast}.
	Therefore, to ensure satisfactory performance, we set $ \tilde M  = 4 $ in our simulation, which requires $ T^{\rm (M-Far)} = \frac{N}{\tilde M}\left(1+\frac{\log _2 \tilde M}{2}\right) $ pilot symbols \cite{you2022fast}.
	\item  {\rm \textbf{Far-field beam training based on DFT codebook:}}
	This scheme utilizes the DFT codebook to select the best codeword for which the maximum received signal power is achieved at the user.
	The beam training overhead of this scheme is $ T^{\rm (Far)} = QM $.
\end{itemize}
\subsection{Performance Comparison}
In this section, we compare our proposed beam training method with benchmark schemes under various setups.
\subsubsection{Effect of reference SNR}
In Fig. \ref{fig:RateSNR}, we plot the achievable rate versus the reference SNR, where the (spatial) angle and range of the typical user is randomly distributed in the region $ [-1,1] $ and $ [10,20] $ m.
Several observations are summarized as follows.
First, it is observed that the proposed multi-beam training method achieves very close rate performance to that of the exhaustive-search method and outperforms other on-grid beam training methods across all SNR regimes (except two-phase beam training method).
This can be explained by the fact that the proposed array-sparse-activation method generates multiple beams without any interference between them. From the perspective of beam pattern, it can achieve the same performance as the exhaustive-search beam training method.
Second, in the low-SNR regime, the proposed multi-beam training method outperforms the two-phase beam training method.
This is because in for the latter scheme the low-SNR regime, the accuracy of angle estimation in the first phase may not be accurate, whose error further propagates to the second phase, thus degrading the overall estimation accuracy.
\begin{figure}
	\centering
	\includegraphics[width=8.2cm]{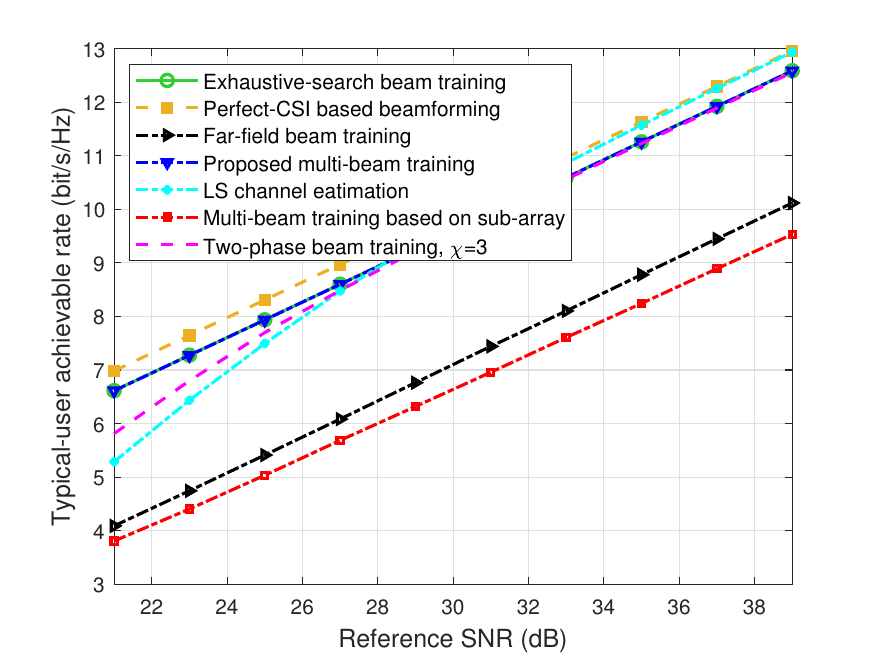}
	\caption{Achievable rate versus reference SNR.}
	\label{fig:RateSNR}
\end{figure}
\begin{figure}
	\centering
	\includegraphics[width=8.2cm]{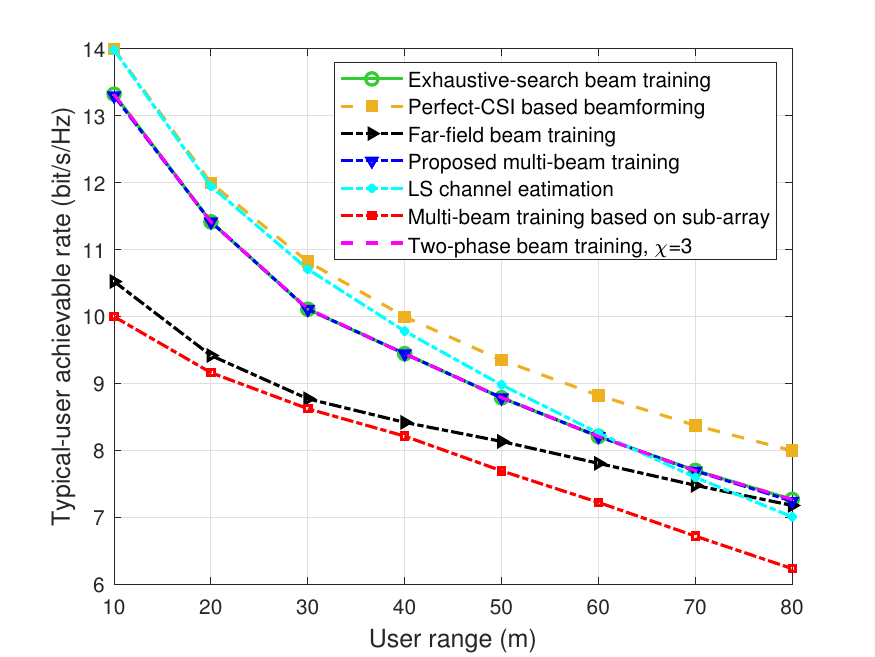}
	\caption{Achievable rate versus user range.}
	\label{fig:RateRange}
\end{figure}
\subsubsection{Effect of user range}
In Fig. \ref{fig:RateRange}, we plot the achievable rate versus the typical user range with $ f = 30 $ GHz, where the spatial angle is randomly distributed in the region $ [-1,1] $ and the transmit power is fixed at $ P_{{\rm tol}} = 30 $ dBm.
It is observed that the proposed multi-beam training method achieves nearly the same performance as the exhaustive-search and two-phase beam training methods.
However,  the proposed method requires much fewer pilot symbols than the other two methods, leading to higher effective achievable rate (to be shown in Fig. \ref{fig:EffectiveRate}).
Moreover, the rate performance of the LS method becomes worse than that of the on-grid near-field beam training methods when the user range increases. This phenomenon arises because the users experience a lower SNR at larger ranges.
\begin{figure}
	\centering
	\includegraphics[width=8.2cm]{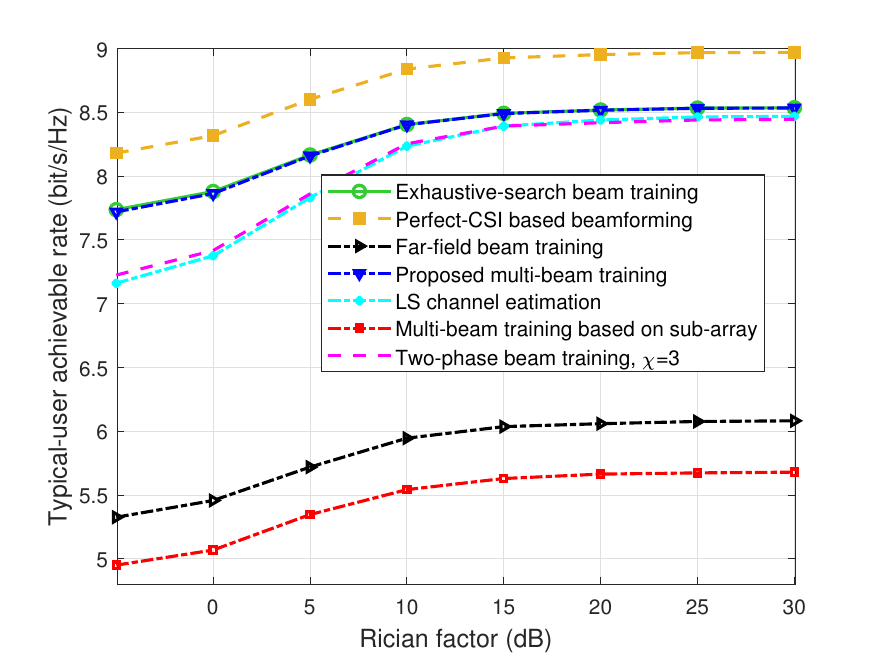}
	\caption{Achievable rate versus Rician factor.}
	\label{fig:RateRician}
\end{figure}
\begin{figure}
	\centering
	\includegraphics[width=8.2cm]{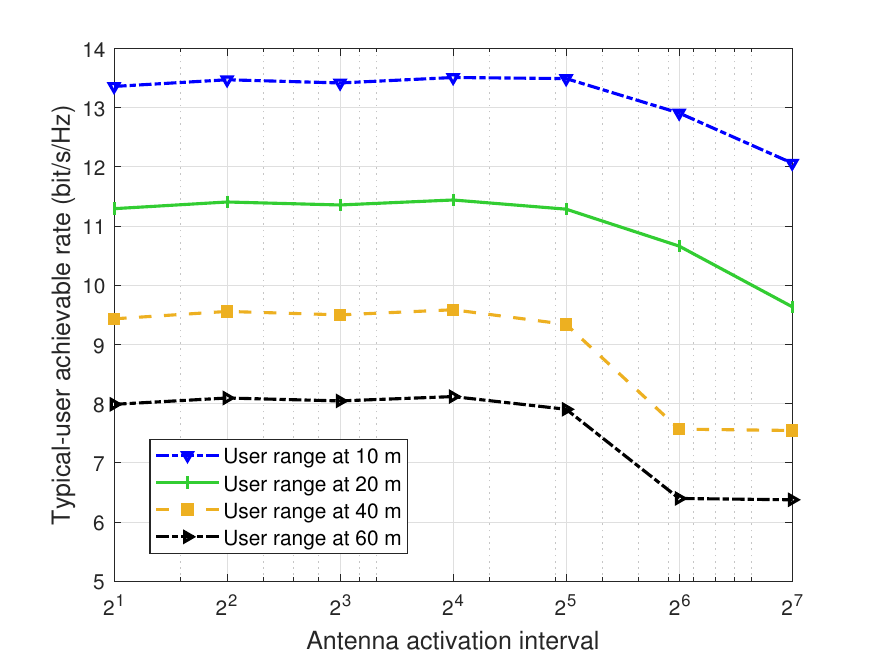}
	\caption{Achievable rate versus antenna activation interval.}
	\label{fig:RateSparsity}
\end{figure}
\subsubsection{Effect of Rician factor}
We plot the achievable rate versus the Rician factor in Fig. \ref{fig:RateRician}, where the reference SNR is fixed at $ 28 $ dB.
The main observations are made as follows.
First, one can observe that the achievable rate of all channel estimation/beam training methods improves with the Rician factor  and it saturates when the Rician factor is sufficiently large.
Moreover, it is observed that the proposed multi-beam training method outperforms the two-phase beam training method in the low-Rician-factor regime.
This can be explained by  two reasons.
On  one hand, the two-phase beam training method imposes stricter requirements on the LoS-dominant channel and high SNR to ensure the accurate user angle estimation in the first stage.
When the Rician factor is small, the NLoS path component degrades the accuracy of angle estimation, whose error further propagates to the second phase, thus compromising the overall beam training performance.
In contrast, the proposed multi-beam training is more robust to the noise, even at the low-SNR regime.
\begin{figure}[t]
	\centering
	\includegraphics[width=8.2cm]{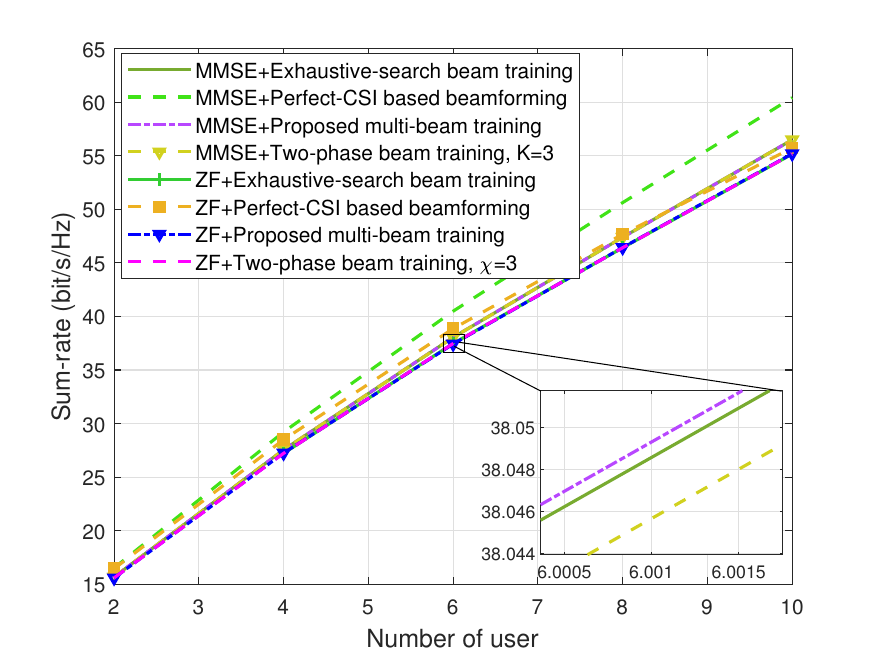}
	\caption{Sum rate versus the number of users.}
	\label{fig:SumRate}
\end{figure}
\begin{figure}[t]
	\centering
	\includegraphics[width=8.2cm]{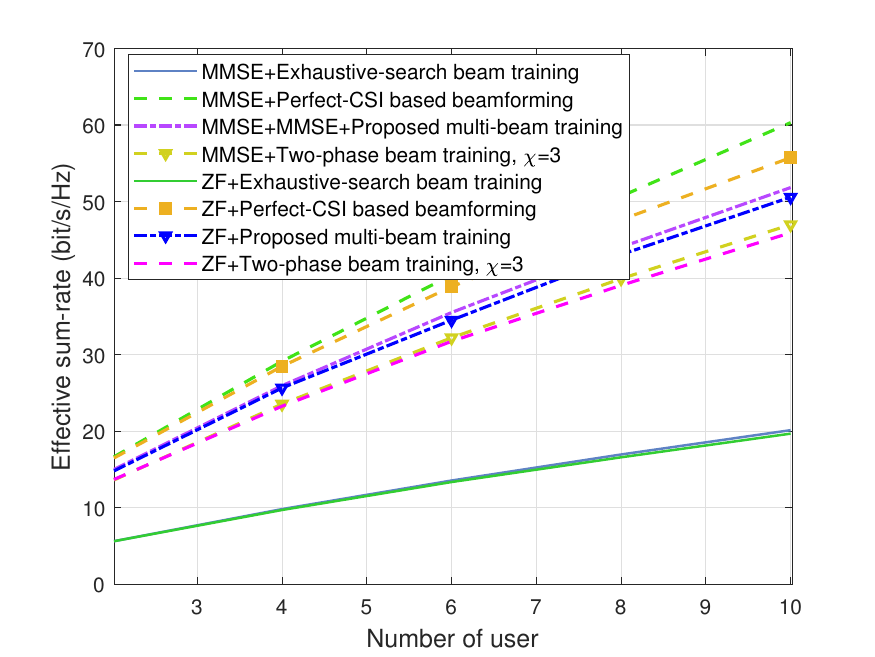}
	\caption{Effective rate versus the number of users.}
	\label{fig:EffectiveRate}
\end{figure}
\subsubsection{Effect of antenna activation interval}
In Fig. \ref{fig:RateSparsity}, we plot the achievable rate versus the antenna activation interval to illustrate the effects of abnormal rings.
It is observed that when the activation interval $ M $ is smaller than $ 32 $ (the optimal solution to $ ({\bf P4}) $), the performance of proposed multi-beam training method remains largely consistent across all user ranges.
However, as $ M $ continues to increase, significant performance degradation occurs.
This is because the conditions for the appearance of abnormal rings  are about $ M > 16  $.
In Fig. \ref{fig:RateSparsity}, when $ M = 64 $ and $ M = 128 $, the Type-III, II, and I rings sequentially appear, resulting in significant performance loss due to the user location ambiguity.
\subsubsection{Effect of number of users}
In Figs. \ref{fig:SumRate} and \ref{fig:EffectiveRate}, we plot the sum-rate and effective sum-rate versus the number of users, respectively.
From Fig. \ref{fig:SumRate},  it is observed that regardless of the number of users and whether the digital beamforming is based on ZF or MMSE, the proposed multi-beam training method always achieves rate performance approaching that of the exhaustive-search method. 
\subsubsection{Effect of frequency band}
\begin{figure}
	\centering
	\includegraphics[width=8.2cm]{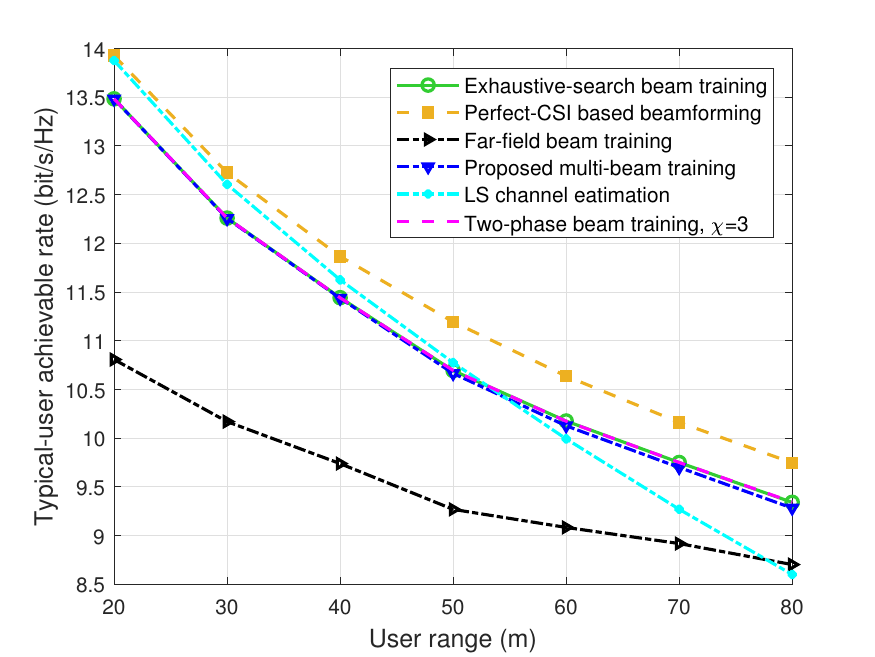}
	\caption{Achievable rate versus user range with $ f = 300 $ GHz.}
	\label{fig:RateRange300GHz}
\end{figure}
\begin{figure}
	\centering
	\includegraphics[width=8.2cm]{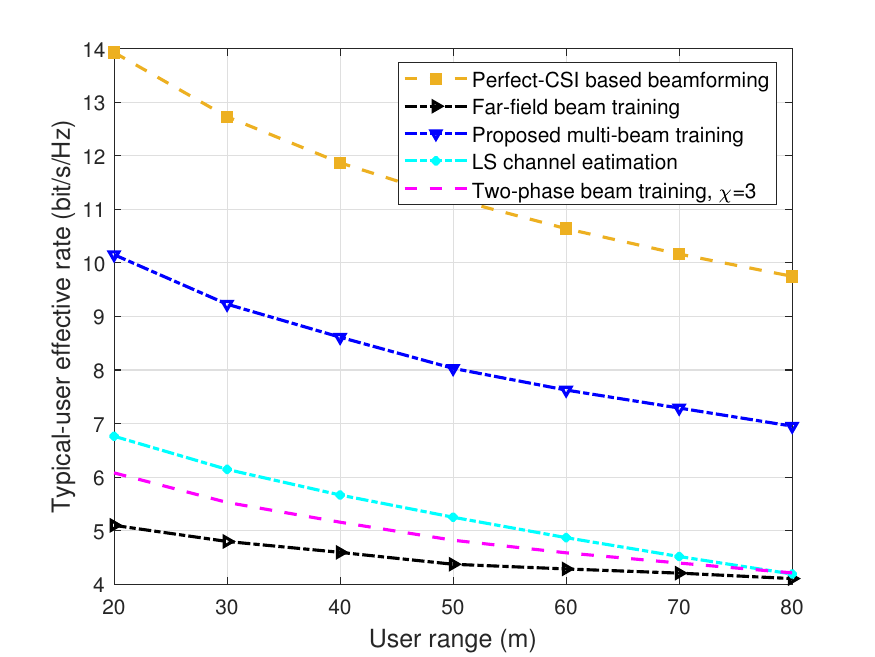}
	\caption{Achievable rate versus user range with $ f = 300 $ GHz.}
	\label{fig:EffRateRange300GHz}
\end{figure}
Finally, we  demonstrate the performance of the proposed multi-beam training method in THz bands, for which the LoS channel model is modified as
\begin{equation}\label{Eq:THz nf-channel}
	\mathbf{h}^H_{k} = \sqrt{N}\beta_{k} \mathrm{e}^{-\frac{1}{2} A(f) r_k} \mathbf{b}^{H}(r_{k}, \theta_{k}),
\end{equation}where $ A(f)$ denotes absorption loss caused by molecules in the propagation medium such as oxygen and carbon dioxide and we set $ A(f) = 5.157 × 10^{-4} $ dB/m \cite{su2023wideband,han2014multi}.
To compensate for severe path loss, the total power and number of antenna are set as $ P_{\rm tol} = 48 $ dBm and $ N = 1025 $, respectively.
We plot in Fig. \ref{fig:RateRange300GHz} the achievable rate versus the user range where $ f = 300 $ GHz.
It is observed that our proposed multi-beam training method still performs very well in the THz band, which achieves similar rate performance as that of the exhaustive-search method.
In addition, as compared to Fig. \ref{fig:RateRange300GHz}, one can observe significant performance gap between the LS channel estimation and other near-field beam training methods. 
This is because the LS channel estimation method is more sensitive to noise, and under the THz circumstance, the received SNR of the user decreases rapidly with the user range due to the severe path loss.
Moreover, we plot the effective rate versus the user range in Fig. \ref{fig:EffRateRange300GHz}.
It is observed that the proposed multi-beam training method achieves significantly better performance as compared with other benchmark schemes due to lower pilot overhead.
\section{Conclusions}
\label{Sec:Conclusion}
In this paper, we proposed a fast near-field multi-beam training method by utilizing a new array-sparse-activation method.
First, we characterized the beam pattern of the proposed array-sparse-activation method and presented sufficient conditions for generating desired multi-beam pattern.
Then, we designed a near-field multi-beam codebook to identifying candidate locations for the user and a single-beam codebook to determine each user location, based on which a new multi-beam training method was proposed to achieve efficient beam training with low beam training overhead.
Finally, numerical results were presented to demonstrate the effectiveness of the proposed multi-beam training method as compared with various benchmark schemes.

\bibliographystyle{IEEEtran}
\bibliography{IEEEabrv}

\begin{thebibliography}{10}
\providecommand{\url}[1]{#1}
\csname url@samestyle\endcsname
\providecommand{\newblock}{\relax}
\providecommand{\bibinfo}[2]{#2}
\providecommand{\BIBentrySTDinterwordspacing}{\spaceskip=0pt\relax}
\providecommand{\BIBentryALTinterwordstretchfactor}{4}
\providecommand{\BIBentryALTinterwordspacing}{\spaceskip=\fontdimen2\font plus
\BIBentryALTinterwordstretchfactor\fontdimen3\font minus
  \fontdimen4\font\relax}
\providecommand{\BIBforeignlanguage}[2]{{%
\expandafter\ifx\csname l@#1\endcsname\relax
\typeout{** WARNING: IEEEtran.bst: No hyphenation pattern has been}%
\typeout{** loaded for the language `#1'. Using the pattern for}%
\typeout{** the default language instead.}%
\else
\language=\csname l@#1\endcsname
\fi
#2}}
\providecommand{\BIBdecl}{\relax}
\BIBdecl

\bibitem{you2024next}
C.~You, Y.~Cai, Y.~Liu, M.~Di~Renzo, T.~M. Duman, A.~Yener, and A.~L.
  Swindlehurst, ``Next generation advanced transceiver technologies for {6G},''
  \emph{arXiv preprint arXiv:2403.16458}, 2024.

\bibitem{lu2023tutorial}
H.~Lu, Y.~Zeng, C.~You, Y.~Han, J.~Zhang, Z.~Wang, Z.~Dong, S.~Jin, C.-X. Wang,
  T.~Jiang \emph{et~al.}, ``A tutorial on near-field {XL-MIMO} communications
  towards {6G},'' \emph{arXiv preprint arXiv:2310.11044}, 2023.

\bibitem{zhang20236g}
H.~Zhang, N.~Shlezinger, F.~Guidi, D.~Dardari, and Y.~C. Eldar, ``{6G} wireless
  communications: From far-field beam steering to near-field beam focusing,''
  \emph{IEEE Commun. Mag.}, vol.~61, no.~4, pp. 72--77, Mar. 2023.

\bibitem{Cui2022channel}
M.~Cui and L.~Dai, ``{Channel estimation for extremely large-scale {MIMO}:
  Far-field or near-field?}'' \emph{IEEE Trans. Commun.}, vol.~70, no.~4, pp.
  2663--2677, Jan. 2022.

\bibitem{zhang2023near}
X.~Zhang, H.~Zhang, and Y.~C. Eldar, ``Near-field sparse channel representation
  and estimation in {6G} wireless communications,'' \emph{IEEE Trans. Commun.},
  vol.~72, no.~1, pp. 450--464, Jan. 2023.

\bibitem{lu2023near}
Y.~Lu and L.~Dai, ``Near-field channel estimation in mixed {LoS/NLoS}
  environments for extremely large-scale {MIMO} systems,'' \emph{IEEE Trans.
  Commun.}, vol.~71, no.~6, pp. 3694--3707, Jun. 2023.

\bibitem{cui2023nearwideband}
M.~Cui and L.~Dai, ``Near-field wideband channel estimation for extremely
  large-scale {MIMO},'' \emph{Sci. China. Inf. Sci.}, vol.~66, no.~7, p.
  172303, Jul. 2023.

\bibitem{wei2021channel}
X.~Wei and L.~Dai, ``Channel estimation for extremely large-scale massive
  {MIMO}: Far-field, near-field, or hybrid-field?'' \emph{IEEE Commun. Lett.},
  vol.~26, no.~1, pp. 177--181, Nov. 2021.

\bibitem{chen2023non}
Y.~Chen and L.~Dai, ``Non-stationary channel estimation for extremely
  large-scale {MIMO},'' \emph{IEEE Trans. Wireless Commun.}, Early Access, Dec.
  2023.

\bibitem{zhang2023nearYou}
Y.~Zhang, B.~Di, H.~Zhang, and L.~Song, ``Near-far field beamforming for
  holographic multiple-input multiple-output,'' \emph{J. Commun. Inf. Netw.},
  vol.~8, no.~2, pp. 99--110, Jun. 2023.

\bibitem{zhang2022dual}
Y.~Zhang, B.~Di, H.~Zhang, M.~Dong, L.~Yang, and L.~Song, ``Dual codebook
  design for intelligent omni-surface aided communications,'' \emph{IEEE Trans.
  Wireless Commun.}, vol.~21, no.~11, pp. 9232--9245, May 2022.

\bibitem{xiao2016hierarchical}
Z.~Xiao, T.~He, P.~Xia, and X.-G. Xia, ``Hierarchical codebook design for
  beamforming training in millimeter-wave communication,'' \emph{IEEE Trans.
  Wireless Commun.}, vol.~15, no.~5, pp. 3380--3392, May 2016.

\bibitem{hassanieh2018fast}
H.~Hassanieh, O.~Abari, M.~Rodriguez, M.~Abdelghany, D.~Katabi, and P.~Indyk,
  ``Fast millimeter wave beam alignment,'' in \emph{Proc. 2018 Conf. of the
  {ACM} Special Interest Group on Data Commun. (SIGCOMM)}, Budapest, Hungary,
  Aug. 2018, pp. 432--445.

\bibitem{zheng2022survey}
B.~Zheng, C.~You, W.~Mei, and R.~Zhang, ``A survey on channel estimation and
  practical passive beamforming design for intelligent reflecting surface aided
  wireless communications,'' \emph{IEEE Commun. Surv. Tuts.}, vol.~24, no.~2,
  pp. 1035--1071, Sept. 2022.

\bibitem{cui2022near}
M.~Cui, Z.~Wu, Y.~Lu, X.~Wei, and L.~Dai, ``Near-field {MIMO} communications
  for {6G}: Fundamentals, challenges, potentials, and future directions,''
  \emph{IEEE Commun. Mag.}, vol.~61, no.~1, pp. 40--46, Sept. 2022.

\bibitem{Zhang2022fast}
Y.~Zhang, X.~Wu, and C.~You, ``{Fast near-field beam training for extremely
  large-scale array},'' \emph{IEEE Wireless Commun. Lett.}, vol.~11, no.~12,
  pp. 2625--2629, Dec. 2022.

\bibitem{wu2024near}
X.~Wu, C.~You, J.~Li, and Y.~Zhang, ``Near-field beam training: Joint angle and
  range estimation with {DFT} codebook,'' \emph{IEEE Trans. Wireless Commun.},
  Early Access, Apr. 2024.

\bibitem{wu2023two}
C.~Wu, C.~You, Y.~Liu, L.~Chen, and S.~Shi, ``Two-stage hierarchical beam
  training for near-field communications,'' \emph{IEEE Trans. Veh. Technol.},
  vol.~73, no.~2, pp. 2032--2044, Feb. 2024.

\bibitem{lu2023hierarchical}
Y.~Lu, Z.~Zhang, and L.~Dai, ``Hierarchical beam training for extremely
  large-scale {MIMO}: From far-field to near-field,'' \emph{IEEE Trans.
  Commun.}, vol.~72, no.~4, pp. 2247--2259, Dec. 2023.

\bibitem{chen2023hierarchical}
J.~Chen, F.~Gao, M.~Jian, and W.~Yuan, ``Hierarchical codebook design for
  near-field mmwave {MIMO} communications systems,'' \emph{IEEE Wireless
  Commun. Lett.}, vol.~12, no.~11, pp. 1926--1930, Nov. 2023.

\bibitem{you2022fast}
C.~You, B.~Zheng, and R.~Zhang, ``Fast beam training for {IRS}-assisted
  multiuser communications,'' \emph{IEEE Wireless Commun. Lett.}, vol.~9,
  no.~11, pp. 1845--1849, Jun. 2020.

\bibitem{cui2023rainbow}
M.~Cui, L.~Dai, Z.~Wang, S.~Zhou, and N.~Ge, ``Near-field rainbow: Wideband
  beam training for {XL-MIMO},'' \emph{IEEE Trans. Wireless Commun.}, vol.~22,
  no.~6, pp. 3899--3912, Jun. 2023.

\bibitem{selvan2017fraunhofer}
K.~T. Selvan and R.~Janaswamy, ``Fraunhofer and fresnel distances: Unified
  derivation for aperture antennas,'' \emph{IEEE Antennas Propag. Mag.},
  vol.~59, no.~4, pp. 12--15, Aug. 2017.

\bibitem{ouyang2024impact}
C.~Ouyang, Z.~Wang, B.~Zhao, X.~Zhang, and Y.~Liu, ``On the impact of reactive
  region on the near-field channel gain,'' \emph{arXiv preprint
  arXiv:2404.08343}, 2024.

\bibitem{zhang2023joint}
Y.~Zhang and C.~You, ``{SWIPT} in mixed near- and far-field channels: Joint
  beam scheduling and power allocation,'' \emph{IEEE J. Sel. Areas Commun.},
  vol.~42, no.~6, pp. 1583--1597, Apr. 2024.

\bibitem{zhang2023mixedYou}
Y.~Zhang, C.~You, L.~Chen, and B.~Zheng, ``Mixed near-and far-field
  communications for extremely large-scale array: An interference
  perspective,'' \emph{IEEE Commun. Lett.}, vol.~21, no.~9, pp. 2496--2500,
  Sept. 2023.

\bibitem{10123941}
Z.~Wu and L.~Dai, ``{Multiple access for near-field communications: {SDMA} or
  {LDMA}?}'' \emph{IEEE J. Sel. Areas Commun.}, vol.~41, no.~6, pp. 1918--1935,
  Jun. 2023.

\bibitem{LuCommunicating2022}
H.~Lu and Y.~Zeng, ``{Communicating with extremely large-scale array/surface:
  Unified modeling and performance analysis},'' \emph{IEEE Trans. Wireless
  Commun.}, vol.~21, no.~6, pp. 4039--4053, Jun. 2022.

\bibitem{liu2023near}
Y.~Liu, Z.~Wang, J.~Xu, C.~Ouyang, X.~Mu, and R.~Schober, ``Near-field
  communications: A tutorial review,'' \emph{IEEE Open J. Commun. Society},
  vol.~4, pp. 1999--2049, Aug. 2023.

\bibitem{sun2019beam}
X.~Sun, C.~Qi, and G.~Y. Li, ``Beam training and allocation for multiuser
  millimeter wave massive {MIMO} systems,'' \emph{IEEE Trans. Wireless
  Commun.}, vol.~18, no.~2, pp. 1041--1053, Feb. 2019.

\bibitem{liu2024near}
W.~Liu, C.~Pan, H.~Ren, J.~Wang, R.~Schober, and L.~Hanzo, ``Near-field
  multiuser beam-training for extremely large-scale {MIMO} systems,''
  \emph{arXiv preprint arXiv:2402.13597}, 2024.

\bibitem{he2017codebook}
S.~He, J.~Wang, Y.~Huang, B.~Ottersten, and W.~Hong, ``Codebook-based hybrid
  precoding for millimeter wave multiuser systems,'' \emph{IEEE Trans. Signal
  Process.}, vol.~65, no.~20, pp. 5289--5304, Oct. 2017.

\bibitem{wu2023enabling}
Z.~Wu, M.~Cui, and L.~Dai, ``Enabling more users to benefit from near-field
  communications: From linear to circular array,'' \emph{IEEE Trans. Wireless
  Commun.}, vol.~23, no.~4, pp. 3735--3748, Sept. 2023.

\bibitem{zhou2024sparse}
C.~Zhou, C.~You, H.~Zhang, L.~Chen, and S.~Shi, ``Sparse array enabled
  near-field communications: Beam pattern analysis and hybrid beamforming
  design,'' \emph{arXiv preprint arXiv:2401.05690}, 2024.

\bibitem{9723331}
E.~Bj{\"o}rnson, {\"O}.~T. Demir, and L.~Sanguinetti, ``A primer on near-field
  beamforming for arrays and reconfigurable intelligent surfaces,'' in
  \emph{Proc. 55th Asilomar Conf. Signals Syst. Comput.}, Nov. 2021, pp.
  105--112.

\bibitem{kosasih2023finite}
A.~Kosasih and E.~Bj{\"o}rnson, ``Finite beam depth analysis for large
  arrays,'' \emph{arXiv preprint arXiv:2306.12367}, 2023.

\bibitem{liu2023low}
W.~Liu, C.~Pan, H.~Ren, F.~Shu, S.~Jin, and J.~Wang, ``Low-overhead beam
  training scheme for extremely large-scale {RIS} in near field,'' \emph{IEEE
  Trans. Commun.}, vol.~71, no.~8, pp. 4924--4940, May 2023.

\bibitem{shi2023spatial}
X.~Shi, J.~Wang, Z.~Sun, and J.~Song, ``Spatial-chirp codebook-based
  hierarchical beam training for extremely large-scale massive {MIMO},''
  \emph{IEEE Trans. Wireless Commun.}, vol.~23, no.~4, pp. 2824--2838, Aug.
  2023.

\bibitem{coleri2002channel}
S.~Coleri, M.~Ergen, A.~Puri, and A.~Bahai, ``Channel estimation techniques
  based on pilot arrangement in {OFDM} systems,'' \emph{IEEE Trans.
  broadcast.}, vol.~48, no.~3, pp. 223--229, Sept. 2002.

\bibitem{su2023wideband}
R.~Su, L.~Dai, and D.~W.~K. Ng, ``Wideband precoding for {RIS}-aided {THz}
  communications,'' \emph{IEEE Trans. Commun.}, vol.~71, no.~6, pp. 3592--3604,
  Mar. 2023.

\bibitem{han2014multi}
C.~Han, A.~O. Bicen, and I.~F. Akyildiz, ``Multi-ray channel modeling and
  wideband characterization for wireless communications in the terahertz
  band,'' \emph{IEEE Trans. Wireless Commun.}, vol.~14, no.~5, pp. 2402--2412,
  Dec. 2014.

\end{thebibliography}

\end{document}